\documentclass[acmsmall,nonacm,acmthm]{acmart}
\usepackage{comment}

\usepackage{float}

\usepackage{amssymb}

\usepackage{xfrac}

\usepackage{changepage}%used by fig3 to arrange 
\usepackage{hyperref}

\AtEndPreamble{%
	\theoremstyle{acmtheorem}%
	\newtheorem{Claim}[theorem]{Claim}%
	\newtheorem{Remark}{Remark}[section]
}

\usepackage{algorithm}
\usepackage{algpseudocode}
\algnewcommand{\LeftComment}[1]{\Statex \(\triangleright\) #1}

\author{Michael Elkin}
\authornote{This research is supported by the ISF grant 3413/25.}
\affiliation{
	\institution{Ben Gurion University of the Negev}
	\city{Beer Sheva}
	\country{Israel}}
\email{elkinm@bgu.ac.il}

\author{Tanya Goldenfeld}
\authornotemark[1]
\affiliation{
	\institution{Ben Gurion University of the Negev}
	\city{Beer Sheva}
	\country{Israel}}
\email{goldenft@post.bgu.ac.il}

\title{Time, Message and Memory-Optimal Distributed Minimum Spanning Tree and Partwise Aggregation}

\begin{document}
 	\begin{abstract}
    \label{0_Abstract}
    
	Memory-(in)efficiency is a crucial consideration that oftentimes prevents deployment of state-of-the-art distributed algorithms in real-life modern networks \cite{BCCHLS25}. In the context of the MST problem, roughly speaking, there are three types of algorithms. The GHS algorithm and its versions \cite{GHS83, Awerbuch87} are memory- and message- efficient, but their running time is at least linear in the number of vertices $n$, even when the unweighted diameter $D$ is much smaller than $n$. The GKP algorithm and its versions \cite{gkp, KP1998fast, elkin2006faster} are time-efficient, but not message- or memory-efficient. The more recent algorithms of \cite {PRS17, E20Simple, haeupler2018round} are time- and message-efficient, but are not memory-efficient. As a result, GHS-type algorithms are much more prominent in real-life applications than time-efficient ones \cite{murmu2015distributed}. In this paper we develop a deterministic time-, message- and \emph{memory-efficient} algorithm for the MST problem. It is also applicable to the more general \emph{partwise aggregation} problem. We believe that our techniques will be useful for devising memory-efficient versions for many other distributed problems.
	\end{abstract}

\maketitle
	\section{Introduction}
	\label{1_Introduction}    

\subsection{Background and Main Results}
\label{1.1_BgMainResults}    
    The Minimum Spanning Tree problem (hereafter, MST) is one of the most fundamental, classical and extensively studied graph problems in distributed computing. The cornerstone work in this area is the paper by Gallager et al. \cite{GHS83}. Their algorithm, called the \textit{GHS algorithm}, requires $O(n \log n)$ time and $O(m+n \log n)$ messages, where $n=|V|$ and $m = |E|$ are the number of vertices and edges, respectively, of the input graph $G = (V,E)$. Their result was improved by \cite{ChinTing85, Gafni85, Awerbuch87,FM04}, ultimately achieving time $O(n)$ and message complexity $O(m+n \log n)$ \cite{Awerbuch87}. These algorithms have an optimal memory complexity of $O(\log n)$\footnote{The memory complexity of a distributed algorithm is the maximum amount of space (in bits) stored in the local memory of any of the vertices.}, in addition to time $\tilde{O}(n)$ and message complexity $O(m + n \log n)$.
    
    Awerbuch et al. \cite{awerbuch1990trade} and Kutten et al. \cite{kutten2015complexity} showed that $\Omega(m)$ messages are required for any deterministic MST algorithm. Moreover, their lower bound is applicable to randomized comparison-based algorithms, and even to randomized algorithms that are not comparison-based, if they operate in the clean network model. (In this model, at the beginning of computation vertices know only their own IDs.) In this paper we focus on this model, on comparison-based algorithms, and mostly on deterministic ones. 
    
    Following a breakthrough by Garay et al. \cite{gkp}, Kutten and Peleg \cite{KP1998fast} devised an algorithm with running time $O(D+\sqrt{n}\log^* n)$ and message complexity $O(m+n^{\sfrac{3}{2}})$. Peleg and Rubinovich \cite{PelegRub99_nearTightLwrBound} proved that the running time of the algorithm of \cite{KP1998fast} is near-optimal. Specifically, they showed that $\Tilde{\Omega}(\sqrt{n})$ time is required for this problem, even when $D = O(\log n)$. Following a breakthrough of \cite{PRS17}, algorithms with near-optimal time \textit{and} message complexities were also devised in \cite{E20Simple, haeupler2018round}. The state-of-the-art algorithm of \cite{E20Simple} is deterministic and requires $O(D + \sqrt n)\cdot \log n$ time and $O((m + n \cdot \log^* n)\cdot \log n)$ messages.
    
	\begin{table*}[h]
	
	\centering
	\label{Tbl1}
	\caption{A concise summary of distributed MST algorithms}
	\begin{tabular}{ c c c c c} 
		\toprule
		Reference & Time & Message & Memory & Computation \\ [0.5ex] 
		\midrule
		Gallager et al. \cite{GHS83} & $O(n \log n)$ & $O(m+ n\log n)$ & $O(\log n)$ & Deterministic \\ 
		Kutten and Peleg \cite{KP1998fast} & $\tilde{O}(D+\sqrt{n})$ & $O(m+n^{\sfrac{3}{2}})$ & $O(\sqrt n)$ & Deterministic \\ 
		Pandurangan et al. \cite{PRS17} & $\tilde{O}(D+\sqrt{n})$ & $\tilde{O}(m)$ & $O(\sqrt n)$ & Randomized  \\
		Elkin \cite{E20Simple} & $\tilde{O}(D+\sqrt{n})$ & $\tilde{O}(m)$ & $O(\sqrt n)$ &Deterministic \\
		Haeupler et al. \cite{haeupler2018round} & $\tilde{O}(D+\sqrt{n})$ & $\tilde{O}(m)$ & $O(\sqrt n)$ & Randomized, Deterministic \\
		\textbf{Our algorithm} & $\tilde{O}(D+\sqrt{n})$ & $\tilde{O}(m)$ & $O(\log^2 n)$ & Deterministic \\
		\bottomrule
	\end{tabular}
\end{table*}
 
However, the \textit{memory complexity} of these algorithms \cite{PRS17, E20Simple, haeupler2018round} is not analyzed, and can in fact be quite large. Unfortunately, this is also the case for the algorithms of Garay et al. \cite{gkp} and of Kutten and Peleg \cite{KP1998fast} (which are time-efficient but not message-efficient).

This leads to the following open problem: 
\\\\
\textbf{Open problem:} \textit{Does there exist an algorithm for the MST problem which is simultaneously near-optimal in terms of time, message and memory complexity?}
\\

We answer this question in the affirmative, and devise a deterministic algorithm with running time $O((D + \sqrt n)\cdot \log^2 n \cdot \log^* n)$, message complexity $O((m + n \log n)\cdot \log n)$ and memory complexity $O(\log^2 n)$. We also show that this result applies to the more general \textit{partwise aggregation} problem (see Section  \hyperref[2.2_PWA]{2.2} for its definition). This problem was introduced by Haeupler et al. \cite{haeupler2016low, haeupler2018faster, haeupler2018round}, and was shown to be a powerful abstraction of many other important classical distributed problems.

\subsection{Related Work}
\label{1.2_Related}

\subsubsection{Memory Complexity}
\label{1.2.1_MemCplx}
Memory efficiency is a central consideration in the context of self-stabilizing distributed algorithms \cite {awerbuchVarghese91, awerbuchOstrovsky1994memory, datta2008self, KKM15_SelfStabilizing}. Self-stabilizing algorithms for MST with small memory complexity were devised in \cite {higham2001self, gupta2003self, tixeuil2009self, blin2013fast, KKM15_SelfStabilizing}. However, all these algorithms require $\Omega(n)$ time. Memory-efficient distributed algorithms were also explored for the closely related problem of \textit{MST verification} \cite{awerbuchVarghese91, korman2007controller, kor2011tight, sarmaEtAl12_distributedVerification, KKM15_SelfStabilizing}. Solutions for this problem typically consist of a pair of distributed algorithms: the \textit{marker} algorithm that assigns labels to vertices, and the \textit{verifier} algorithm that uses these labels to validate whether the input subgraph is an MST. Korman et al. \cite{KKM15_SelfStabilizing} demonstrated that both these tasks can be efficiently accomplished using small memory. 

Another distributed problem for which memory efficiency is crucial is the \textit{routing} problem \cite{awerbuch1990improved, awerbuch1992routing, gavoille2013communication,lenzen2013fast, lenzen2015fast, elkin2018near, elkin2018efficient, lenzen2019distributed}. Similarly to the MST verification problem, the routing problem also involves two stages. The first is a  preprocessing stage, during which routing tables and labels are computed by a distributed algorithm. Next, there follows a routing stage. In \cite{elkin2018efficient, elkin2018near, lenzen2019distributed} both these tasks are implemented efficiently, while using small memory.

Recently Ben Basat et al. \cite{BCCHLS25} explored memory efficiency in the context of clique-listing algorithms. They motivate the importance of designing memory-efficient algorithms by the need to deploy them in modern communication environments, where vertices are severely limited in terms of their computational resources.

Emek and Wattenhofer \cite{emek2013stone} considered an even more memory-restrictive model than ours, or than any of the aforementioned papers. While we allow polylogarithmic memory, they require \textit{constant} space complexity. Their model is motivated by wireless and ad-hoc networks, as well as by networks of biological cells. To be meaningful, this model requires randomization, and they indeed devise very elegant and efficient randomized algorithms for a number of core symmetry-breaking problems, such as MIS and 3-coloring of oriented trees. 
Memory-efficient algorithms for distributed MST were also devised in the MPC model. See \cite{azarmehr2025massively}, and the references therein.

\subsubsection{Additional Related Work}
\label{1.2.2_MST}
We refer to an excellent survey of Pandurangan et al. \cite{P18SURVEY} for a thorough overview of the enormous body of research concerning the distributed MST problem. Randomized non-comparison-based MST algorithms that apply to a slightly less stringent $KT_1$ model (in which at the beginning of the computation every vertex knows the IDs of its neighbors) with message complexity $\tilde{O}(n)$ were devised in \cite{king2015construction, mashreghi2021broadcast, gmyr2018time}. Their time and memory complexities are, however, quite large. Approximation algorithms and hardness of approximation of distributed MST was studied in \cite{elkin2006unconditional, sarmaEtAl12_distributedVerification, khan2008fast, elkin2014can}. A time- and message-efficient (but not memory-efficient) distributed MST algorithm that works in the asynchronous model was devised by Dufoulon et al. \cite{dufoulon2022almost}.

\subsection{Technical Overview}
\label{1.3_TechOvViw}

The algorithms of Garay et al. \cite{gkp} and of Kutten and Peleg \cite{KP1998fast} (that achieve sublinear time in $n$ for $D = o(n)$) compute a significant part of the solution using a supergraph in which fragments of the MST constitute vertices. Vertices upcast messages about the edges of this supergraph over an auxiliary BFS tree $T$, which is computed at the preprocessing stage of these algorithms. As a result, vertices need to store information about $\Omega(\sqrt n)$ fragments.

This high memory requirement is inherited by the first two time- and message-optimal MST algorithms, which are due to Pandurangan et al. \cite{PRS17} and Elkin \cite{E20Simple}. Roughly speaking, the algorithms of \cite {gkp, KP1998fast, elkin2006faster, PRS17, E20Simple} consist of two stages. First, a version of the GHS algorithm is executed, up until fragments become too large. At this point, the fragments are called \textit{base fragments}, and the algorithms process them in the second stage. In both the algorithms of \cite{PRS17, E20Simple}, the root $r$ of $T$ stores some information about each of the base fragments. As the number of base fragments may be $\Omega(\sqrt n)$, this is also a lower bound on the memory complexities of these algorithms.

The algorithm of Haeupler at al. \cite{haeupler2018round} does not suffer from this problem. In this algorithm the base fragments (and the fragments that are created from them) continue running a version of GHS. However, when fragments search for minimum-weight outgoing edges (henceforth, MWOEs), communication is no longer performed over the MST fragments themselves, but rather it is carried out over a \textit{shortcut structure} \cite{haeupler2018round, ghaffari2016distributed, haeupler2018faster, ghaffari2021low, haeupler2016low}. For general graphs, the worst-case shortcut structure reduces to the auxiliary BFS tree $T$ that was discussed above. 

While this reduces the load placed on the root $r$ of $T$, there are other memory-consuming parts in the algorithm of \cite{haeupler2018round}. In particular, all fragments interconnect via subtrees of the same BFS tree $T$, which may overlap. In the worst case, this causes large congestion, of order $\Theta(\sqrt n)$. As a result, vertices need to store up to $\Theta(\sqrt n)$ enqueued messages in their local memory. Haeupler et al. \cite{haeupler2018round} provide a randomized solution to this problem, ensuring that buffers will be of size $O(\log n)$ at all time, with high probability. The deterministic version of their algorithm does, however, suffer from this problem of congestion. 

Secondly, the algorithm of \cite{haeupler2018round} handles broadcasts as symmetrical to convergecasts. To broadcast a message from the root $r_F$ of a fragment $F$ to its vertices $v^F_1, v^F_2, \ldots$, the message needs to pass through vertices in the shortcut structure. All intermediate vertices on the path $r_F$ - $v^F_i$ need to remember routing information that is specific to $F$. A vertex needs to remember a different piece of information for every overlapping shortcut structure that it is a part of. As a result, the memory complexity of the algorithm of \cite{haeupler2018round} (both of its deterministic and randomized versions) is $\Tilde{O}(\sqrt n)$.

To overcome these hurdles, we devise a deterministic time-, message- and memory-efficient \emph{Communication Cycle protocol}, that enables multiple vertex-disjoint subgraphs to communicate with one another in a pipelined manner. The execution of the second part of the algorithm proceeds in phases, mirroring the phases of the GHS algorithm. At the beginning of each phase, our algorithm distributes \textit{time slots} to communicating fragments. Our analysis guarantees that as long as communication cycles employ these time slots, no congestion occurs. (This is achieved \textit{deterministically}, while the logarithmic bound on the congestion of \cite{haeupler2018round} holds with high probability!)
 
The communication cycle protocol maintains message efficiency by sending broadcast messages through a path from the root $r_i$ of a fragment $F_i$ to the BFS tree root $r$, and from $r$ to the roots of the base fragments $B_{i_1},B_{i_2},\ldots$ contained in fragment $F_i$. Analogously, convergecast outcomes are sent from the root $r$ to the root $r_i$ of each fragment $F_i$. A similar approach was used in \cite{E20Simple}, but here we adapt it to the more stringent memory-restrictive model. To direct messages from $r$ to the correct fragment roots, or base fragment roots, we employ Thorup-Zwick's tree routing mechanism \cite{TZ01Compact}. During preprocessing, we build and distribute tree routing tables and labels. For every slot set, once slots are in place, we use the communication cycle protocol to distribute the routing labels of the roots of communicating fragments. As a result, we are able to route messages via designated shortest $r$-$r_i$ paths, without increasing either memory or message complexity.

Another challenge arises as the root $r$ of $T$ cannot store messages designated to multiple fragments. Consider a situation when we have multiple fragments $F_1, F_2,\ldots,F_t$ rooted at vertices $r_1,r_2,\ldots,r_t$, respectively, and each fragment $F_i$ consists of base fragments $B_{i_1}, B_{i_2}, \ldots$ rooted at vertices $r_{i_1},r_{i_2},\ldots$, respectively. Now each root $r_i$ of $F_i$ ($i\in [t]$) needs to broadcast a message $m_i$ to all roots $r_{i_1},r_{i_2},\ldots$ of the base fragments that comprise $F_i$. The solution employed by \cite{elkin2006faster, PRS17, E20Simple} is to first perform a pipelined convergecast of all messages $m_i$ to the root $r$, and then broadcast these messages to the respective roots of base fragments. This approach, however, is not memory-efficient, as the root $r$ will need to store all these messages before it can dispatch any of them. 

Our algorithm, on the other hand, interleaves convergecasts and (targeted) broadcasts through a carefully designed scheduling mechanism. The scheduling is implemented by arranging time slots in so-called \textit{publish-query} sets. The schedule ensures that immediately after the root $r$ receives message $m_1$ from the root $r_1$ of fragment $F_1$, $r$ obtains the queries from the roots $r_{1_1},r_{1_2},\ldots$ of the base fragments $B_{1_1},B_{1_2},\ldots$ that comprise $F_1$. Crucially, these queries contain the base fragments' routing labels, and thus the root $r$ does not need to store them. Upon receiving a query, the root $r$ sends the message $m_1$ to the base fragment that has sent this query. Once queries from base fragments that comprise $F_1$ are exhausted, the root $r$ removes $m_1$ from its local memory. Only then does the root $r$ receive the message $m_2$ from $r_2$, and afterwards queries from $r_{2_1},r_{2_2},\ldots$, and so on. As a result, the root $r$ is never required to store more than one message in its local memory.

\subsection{Structure of the Paper}
\label{1.4_StructPaper}
In Section \hyperref[2_Prelims]{2} we give basic definitions, and explain some concepts that will be used later. In Section \hyperref[3_Convergecast]{3}, we provide a description of the convergecast variant of the communication cycle protocol. In Section \hyperref[4_Broadcast]{4}, we describe the more complicated broadcast variant of the protocol. In Section \hyperref[5_MSTConstruction]{5} we show how the communication cycle mechanism is used to carry out MST construction. Appendix \hyperref[Appendix.A_ExistingAlgs_MemCplx]{A} provides detailed analysis of the memory complexity of elements of MST construction algorithms used in this work. Appendix \hyperref[Appendix.B_Analysis]{B} provides the details of analysis for the communication cycle procedure.
\section{Preliminaries}
\label{2_Prelims}

\subsection{The Computational Model}
\label{2.1_The computational Model}

We consider the standard CONGEST clean-network (aka $KT_0$) model (see \cite{peleg2000distr}, Ch. 5), but every vertex $v$ has memory that is polylogarithmic in $n$ (specifically, $O(\log^2 n)$ bits). It also has $O(deg(v))$ read-only memory for storing its neighbors and outgoing ports leading to those neighbors. A vertex also has buffers in each of its ports, which store incoming messages. A vertex is allowed to read these messages as long as they are not overwritten by more recent messages, and as long as its local memory is not full.

In this model, in order to store an MST, or in fact any subgraph $H = (V,E_H)$ of the input graph $G=(V,E)$, we require that every edge $e\in E_H$, $e=(u,v)$ is recorded by either endpoint $u$ or $v$ as belonging to $H$. In our algorithm, every vertex ends up storing up to $O(\log n)$ edges of the MST, even though its degree in the MST may be much higher.

This model is the same as in \cite{KKM15_SelfStabilizing, lenzen2013fast,lenzen2015fast, elkin2018efficient, elkin2018near}. Some other works, such as \cite{BCCHLS25}, allow at least $\tilde{O}(deg(v))$ memory per vertex. We note that previous time and message-optimal algorithms for MST \cite{PRS17, E20Simple,haeupler2018round} (and also time-optimal algorithms with higher message complexity \cite{gkp, KP1998fast}) require $\Omega(\sqrt n)$ memory complexity even on graphs with constant maximum degree.

\subsection{The Partwise Aggregation Problem}
\label{2.2_PWA}
Partwise aggregation (hereafter, PWA), formulated in \cite{haeupler2018round}, is a generally stated primitive of distributed computation, which appears in numerous distributed problems. The input of the problem is a graph $G = (V,E)$, a partition $\mathcal{P}=\{P_i\}_{i=1}^{i=N}$ of $V$ such that every part $P_i$ induces a connected subgraph, and a commutative associative function $f$. Every $v\in V$ has some input $x_v$ for function $f$. The problem is considered solved when, for every part $P_i$, the aggregate of function $f$ over all inputs of vertices $v\in P_i$ is calculated, and known to each vertex $v\in P_i$. Finding the MWOEs of fragments can be viewed as an instance of PWA. Every vertex $v$ in a fragment $F$ submits its local minimum-weight edge, and the minimum function is calculated, to find each fragment's minimum.

\subsection{MST-construction Algorithms}
\label{2.3_MSTAlgorithms}

We describe aspects of existing algorithms for constructing MST which we use in this paper.  
\subsubsection{The GKP phase}
\label{2.3.1_GKP}

The GHS algorithm \cite{GHS83} is a parallelization of the Bor\r{u}vka algorithm. We describe it in detail and establish its memory complexity in Appendix \hyperref[Appendix.A.2.1_GHS]{A.2.1}. The modified, more time-efficient version of the GHS phase devised by Garay et al. \cite{gkp} increases both the minimum fragment size and the maximum fragment depth by constant factors. It requires an additional step, after finding the MWOEs. Virtual trees, whose nodes are fragments and whose edges are MWOEs, are partitioned into subtrees, and the subtrees merge. The additional step entails executing the following distributed algorithms, over the virtual trees:
\begin{enumerate}
	\item Computing a proper 3-coloring.
	\item Computing a maximal matching, using the coloring.
	\item Unmatched fragments find an adjacent matched pair and join it.
\end{enumerate}

Once these steps are complete, every matched pair and its attached fragments form a subtree of depth $O(1)$, and every fragment belongs to a subtree of size greater than one. 

We state the following claim without proof:

\begin{Claim}
	\label{C2.1_GKP_MsgTypes}
	\cite{gkp} The messages exchanged between fragments during this additional step belong to one of three types:
\begin{enumerate}
	\item Each parent sends a uniform message to all its children.
	\item Each parent sends a message to one arbitrary child.
	\item Each child fragment messages its parent fragment.
\end{enumerate}
\end{Claim}

The memory complexity of the GKP phase is $O(\log n)$ (see Appendix \hyperref[Appendix.A.2.2_GKP]{A.2.2}).
\subsubsection{Lenzen's phase}
\label{2.3.2_LenzenGKP}
In the GKP phase, the minimum size of fragments grows by a factor of at least 2, while the maximum diameter of fragments grows by a factor of at most 3. The version of the GKP phase described by Lenzen \cite{Lenzen2016}, Ch. 6 (see also Section 4 of \cite{E20Simple}) ensures that both these factors are equal to 2. Time- and message-efficient algorithms of \cite{E20Simple, PRS17, haeupler2018round} rely on Lenzen's routine, to ensure that at the end of the first stage of the algorithm there are at most $O(\min(\sqrt{n},\sfrac{n}{D}))$ base fragments, with maximum diameter $k=O(\max(\sqrt{n},D))$. We also use Lenzen's phase in the first stage of our algorithm to guarantee these bounds of $O(\min(\sqrt{n},\sfrac{n}{D}))$ and $O(\max(\sqrt{n},D))$ on the number of base fragments and their maximum diameter, respectively.
The main difference between Lenzen's phase and the GKP phase is that in the former only fragments with small diameter compute MWOEs and request to merge with their MWOE-adjacent neighbors. (Fragments of large diameter may accept merge requests from neighbors, though.) As we argue in Appendix \hyperref[Appendix.A.2.3_Lenzen]{A.2.3}, Lenzen's phase can be readily implemented with small memory. 

In the second stage of our algorithm we employ the simpler GKP phase (as opposed to Lenzen's phase), as at that point we no longer need to guarantee that the product of the number of fragments and their maximum diameter is linear in $n$. 

\subsubsection{Message-Efficient Algorithms}
\label{2.3.3_MsgEfficient}
The deterministic time- and message-efficient algorithm devised by \cite{E20Simple} (see also \cite{PRS17} for an earlier randomized algorithm) guarantees nearly-tight upper bounds on both time and message complexities. The algorithm executes Lenzen phases until it reaches a point of transition, which is selected depending on both graph diameter $D$ and vertex number $n$. We fix a parameter $k=\max {\{D,\sqrt{n}\}}$. Bor\r{u}vka-style fragment-merging phases are carried out directly over the MST-fragments, until fragment diameter is $O(k)$. The size of fragments at this point is $\Omega(k)$ and so the number of fragments is $N_b=O(\sfrac{n}{k})$. Fragments are then fixed, and thereafter they are called \textbf{base fragments}. Next, Bor\r{u}vka-style fragment-merging phases continue, but communication occurs in two parts - within base fragments, and between base fragments. We state the following claim without proof. See \cite{E20Simple} for more details:
\begin{Claim}
	\label{C2.4_E20_MsgCplx}
	\leavevmode
	\begin{enumerate}
		\item The message complexity of the first part of a phase (i.e., within base fragments) is $O(n)$.
		\item 	The message complexity of the second part of a phase (between base fragments) is $O(D\cdot N_b)=O(n)$.
	\end{enumerate}
	
\end{Claim}

\subsection{Generalizing MST to MSF}
\label{2.4_toMSF}
The MST problem can be stated as a special case of the MSF (minimum spanning forest) problem. Given a graph $G=(V,E)$, a weight function $w(e)$, and a subgraph $H=(V,E_H)$ such that $E_H\subseteq E$, find a forest of minimum spanning trees that span the connected components of $H$. Our algorithm can be extended to solve the more general MSF case. In order to adapt Bor\r{u}vka-style fragment-merging solutions, in general, we substitute weight function $w(e)$ for a modified weight function $w_H(e)$ \footnote{We assume that $E_H$ is stored compactly, so that any $e\in E_H$ is known as such to one endpoint. Any endpoint $v$ of edge $e=(u,v)$ can locally derive $w_H(e)$ after one round of communication with $u$. Thus, no additional memory complexity is required to store $w_H(E)$.}:

\begin{equation}
\label{eqn_subgraph_weights}
	w_H(e) = 
\begin{cases}
	w(e) & e\in E_H\\
	\infty & e\notin E_H
\end{cases}
\end{equation}

If $H$ has more than one connected component, we are required to handle the case where a fragment $F$ reaches its maximum size, and cannot find MWOE-candidates belonging in $E_H$. If this is the case, then the fragment $F$ is required to stop searching for MWOEs. Once all fragments have terminated in this way, the algorithm is completed, and the remaining terminated fragments are the forest of minimum spanning trees.

Note that storing the input graph $H$ may require vertex memory above the polylogarithmic upper bound that we seek to establish for our algorithm. For the sake of our analysis, we assume that this memory is available in read-only regime. However, even though we solve the PWA problem via a reduction to the MSF problem, the resulting algorithm for the PWA problem does not require vertices to store any additional subgraph $H$ (see Corollary \hyperref[C5.9.1_PWA]{5.9} for details).

\section{Performing a Convergecast through a Communication Cycle}
\label{3_Convergecast}
In this section we describe the \textbf{communication cycle}, a procedure that enables each one of $K$ disjoint fragments $F_1,\ldots,F_K$ to perform a convergecast calculation on a set of inputs distributed among its vertices. Each fragment $F_i$ will then send the outcomes to some vertex $r_i\in F_i$ designated as $F_i$'s leader.

Suppose that we have $K$ disjoint fragments $F_1, \ldots, F_K$, and that each fragment $F_i$ is partitioned into connected subsets $\mathcal{B}^i=\{B^i_1,B^i_2,\ldots\}$, $\bigcup_{B\in \mathcal{B}^i}{B} = F$, $ \bigcap_{B\in \mathcal{B}^i}{B} = \emptyset$. Every subset $B$ has a spanning tree $T_B$ with diameter at most $D_b$. We denote the overall number of such connected subsets in $G$ from all fragments $K_b = |\bigcup_{i=1}^K{\mathcal{B}^i}|$. The communication cycle procedure enables all fragments to perform a convergecast calculation on the respective inputs of their vertices. Next, the procedure enables all fragments to deliver their respective outcomes to designated leader vertices. The communication cycle requires $O((D+D_b+K)\log n)$ rounds and $O((D_b\cdot K)\log n)$ messages to carry out, and makes use of $O(\log^2 n)$ bits of memory for each vertex $v\in V$. 

\subsection{Prerequisites of a Convergecast Communication Cycle}
\label{3.1_preq}
Within the communication cycle, the convergecast calculation is carried out in two steps. In the first, \textbf{intrinsic} step, the subsets hold convergecast calculations simultaneously, over their disjoint spanning trees. In the second, \textbf{extrinsic} step, subset roots participate in a convergecast, in a pipelined fashion, over a BFS tree $T$ spanning the entire graph $G$. 

In order to carry out the intrinsic step, each subset $B$ needs a spanning tree $T_B$. For the extrinsic step, we require each vertex $v\in V$ to know its depth $d(v)$ in $T$, and $T$'s overall depth $d(T)$. All vertices discover these values during the computation of $T$. This calculation requires $O(D)$ rounds and $O(n)$ messages. Additionally, in order to carry out the extrinsic step, every fragment $F_i$ is required to be associated with a distinct number $i\in [1,K]$. We will refer to this number as a \textbf{time slot}, or \textbf{slot number}. The vertices contained in $F_i$ are required to know the value of $i$. For a vertex $v\in F_i$, its value $i$ identifies the fragment, and also indicates the position of $F_i$ in the pipeline of convergecast calculations.

In the course of the extrinsic step, for each fragment $F_i$ the final result of the convergecast is calculated at the root of $T$, and routed to $F_i$'s designated leader vertex $r_i$. In order to accomplish this, we use Thorup-Zwick's tree routing scheme (see Appendix \hyperref[A.5.2_TZ_Tree_Routing]{A.5.2}). We require that at the outset, for every fragment $F_i$, every vertex in $F_i$ knows $r_i$'s routing label, $lbl_i$.

We note that labels have size $O(\log^2 n)$ \cite{TZ01Compact}. In the extrinsic step, we send routing labels as part of messages. Because we are in the CONGEST setting, a message carrying a label would require $O(\log n)$ rounds to deliver, and $O(\log n)$ messages. We therefore multiply time complexity and message complexity by a factor of $\log n$. See Remark \hyperref[B.1.Remark_Convt_MemCplx_wrt_storingPartialResults]{B.1} in Appendix \hyperref[Appendix.B_Analysis]{B} for more details.

\subsection{Description of a Single Convergecast Cycle}
\label{3.2_description}
The intrinsic step takes place within the subsets. Subsets simultaneously hold convergecast calculations over their respective disjoint spanning trees. The intrinsic step requires at most $D_b$ rounds, and $O(n)$ messages overall. At the end of the intrinsic step, in every subset $B$, the root vertex $r_B$ of spanning tree $T_B$ of $B$ knows the partial convergecast outcome over the inputs from vertices of $B$. 

The extrinsic step is carried out over the BFS tree $T$ of the entire graph. The extrinsic step requires $2d(T)+K$ rounds to execute. During this time, convergecast calculations are held over $T$ for all $K$ fragments, in a separated and pipelined fashion. For every fragment $F_i$, root vertices of its subsets $B^i_1,B^i_2, \ldots \in\mathcal{B}_i$ inject their partial results into the appropriate convergecast calculation. 

We use time slot numbers to separate the calculations of different fragments into distinct tree depths. Separation of time slots enables us to avoid congestion, and the memory costs of queuing. Separation also removes the need for vertices to store partial results for different fragments.

Each fragment $F_i$ with a slot number $1\leq i\leq K$ has a \textbf{window of exclusive ownership} which travels up the BFS tree $T$. For every round $t\in\{i,i+1,\ldots,i+d(T)\}$, vertices at depth: 
\begin{equation}
	\label{eqn_slot_to_depth}
	dp(i,t)=d(T)+i-t
\end{equation}
handle only messages from fragment $F_i$ on this round, i.e., are \textbf{owned} by $F_i$ on this round. In particular, on round $t$, vertices at depth $dp(i,t)$ receive only messages associated with $F_i$. Any vertex $v$ at some depth $d(v)$ assumes that all the messages that it receives in round $t$ come from a single fragment $F_i$ associated with slot $i=d(v)-d(T)+t$. Thus, the vertex handles all messages together, aggregating them according to the convergecast function designated for the cycle. If a vertex $v$ with depth $dp(i,t)$ belongs to the fragment $F_i$, it adds any input it has to the calculation during round $t$. It then upcasts the result, in the same round. It follows that in round $t$, $dp(i,t)$-depth vertices also send only messages for fragment $F_i$. 

Finally, in round $i+d(T)$, the outcome for $F_i$ reaches $r_T$, the root of $T$ (note that $dp(i,i+d(T))=0$). The outcome is then sent from $r_T$ to $F_i$'s leader vertex $r_i$. The path is encoded in the label $lbl_i$, which is also delivered to $r_T$ during that round of the convergecast. A copy of the label $lbl_i$ is attached to every partial outcome message sent by a subset root, for every subset $B^i_j\subseteq F_i$.

Overall, every slot number $i$ is associated with a sequence of pairs $$\{\langle i+t,d(T)-t \rangle|t\in[0,d(T)]\}$$ of round and depth. This sequence describes the movement of the window of ownership from maximum depth $d(T)$ to tree root $r_T$. A corresponding sequence of pairs $\langle i+t,t-d(T)\rangle,t\in[d(T)+1,2d(T)]$ describes the downward movement of the returning message (see Fig. \hyperref[Fig1_SlotsPipelining]{1}).

\begin{figure}[h]
	\label{Fig1_SlotsPipelining}
	\begin{center}
		\includegraphics[width=\textwidth]{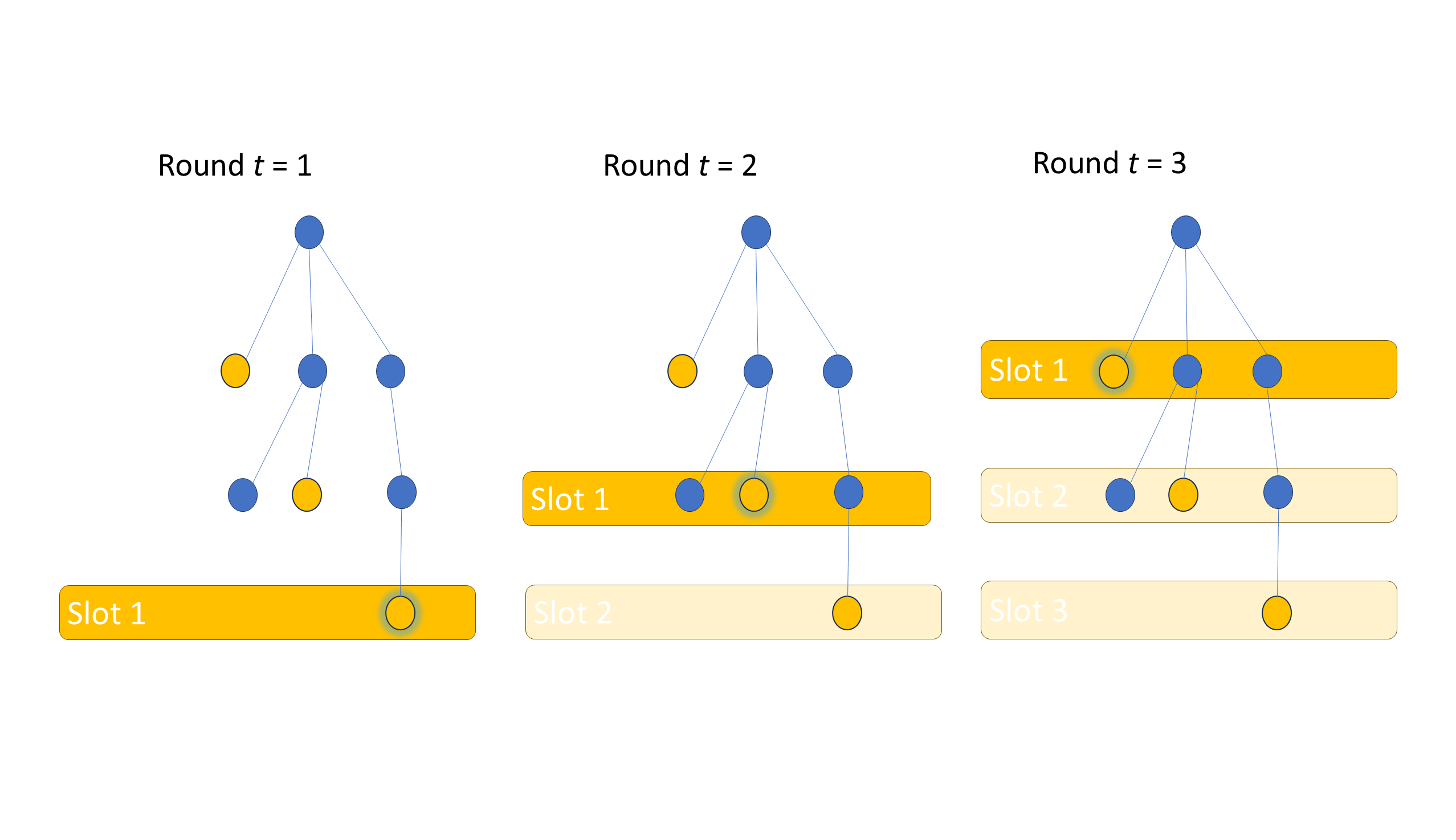}
		
		\caption[Progression of time slots through the BFS-Tree]{Progression of time slots through the BFS Tree $T$: The yellow vertices belong to the fragment $F_1$ associated with slot $1$. The overall tree depth is $d(T)$ (here $d(T)=3$). For each slot $i=1,2,\ldots$, during round $t$ the vertices at depth $d(T)-t+i$ upcast the messages of slot $i$. If a vertex $v$ is associated with slot $i$ and has depth $d(T)-t+i$, then $v$ injects its input into the convergecast during round $t$.}
		
		\Description{Three diagrams showing the same rooted tree graph, with labels indicating consecutive rounds 1, 2, 3. The tree has overall depth 4. At level 0 there is one root vertex, and at levels 1, 2, 3 there are several vertices. One vertex at each level is marked yellow. A yellow bar labeled ``slot 1" is drawn behind the tree in each of the three diagrams, indicating the vertices at depth 3, 2, 1 at rounds 1, 2, 3 respectively. In each diagram, the marked vertex at the depth level indicated by the bar is surrounded with a blue aura.}
	\end{center}
\end{figure}

A vertex $v\in V$ performs two tasks. As a vertex in the BFS tree $T$, it carries out the pipelined convergecasts. If in round $t$ vertex $v$ received input values from its $T$-children, then in that same round $v$ sends the output calculated from these values by the convergecast function $f$. If $v$ is also the root vertex of a subset $B\subseteq F_i$, then $v$ has another responsibility - to inject $B$'s partial convergecast result into the calculation at the correct round, $d(T)-d(v)+i$. In Algorithm \hyperref[Alg1_convergecast]{1} we provide the pseudocode.

\begin{algorithm}
	\label{Alg1_convergecast}
	\caption{Convergecast cycle, a single round of the extrinsic step}
	\begin{algorithmic}[1]
		\Require Vertex $v$ belongs to a base fragment $B\in \mathcal{B}_i$. Convergecast function is $f$

		\LeftComment{Vertex $v$ initializes values $output$, $lbl$:}
		\If{Current round $t$ equals $d(T)-d(v)+i$ AND $v$ is the root vertex of $B$}
		\State Set $output$ to the partial outcome of $B$; set $lbl$ to $lbl_i$.
		\Else
		\State Set $output$ to $\{\}$; set $lbl$ to $\{\}$.
		\EndIf
		
		\LeftComment Vertex $v$ aggregates values from $T$-children:
		\For{each input message from child $u$ received $\langle inp_u,lbl_x\rangle$}
		\If{$inp_u\neq\{\}$}			
		\If{$output \neq  f(output,inp_u)$}
		\State Update $output \coloneq f(output,inp_u)$
		\State Update $lbl\coloneq lbl_x$.
		\EndIf
		\EndIf		
		\EndFor
		
		\State Send message $\langle output, lbl\rangle$ to $T$-parent of $v$.
		
	\end{algorithmic}
\end{algorithm}

Our analysis shows that the algorithm guarantees that, for each fragment $F_i$, all messages originating from $F_i$ are handled by vertices at a particular depth of $T$ at every round. It also ensures that for every round, the messages of different fragments are handled at different depths of $T$ (see Appendix \hyperref[B.1_analysis]{B.1}, Claim \hyperref[CB.1_CycleSeparatesFrags]{B.1}).

The messages that are upcast during the extrinsic step of the communication cycle consist of two items of information. The first item is a result of partial convergecast, and the second item is a routing label. The messages associated with fragment $F_i$ are partial outcomes from the roots of subsets contained in $F_i$, or are aggregated from partial outcomes. The roots of subsets contained in $F_i$ have the routing label of $F_i$'s designated leader vertex $r_i$, and append this label to messages that they compose. In the MST construction algorithm that will be described in Section \hyperref[5_MSTConstruction]{5}, the leader will be the root vertex of some subset within $F_i$. For the purpose of the convergecast cycle, it can be any vertex in $V$, and we require only that the label is consistent across $F_i$. Any vertex $v$ that aggregates outcomes from messages received in this round will copy the label from one of these messages.

Once the final aggregated message $msg_i$ for $F_i$ has reached $T$-root $r_T$, it is routed from $r_T$ to $r_i$ according to $lbl_i$. This path requires up to $d(T)$ steps and up to $d(T)$ messages. All messages are sent down the tree $T$, in consecutive rounds, and do not create congestion with upcasts, or with one another. When a convergecast cycle is completed, for each fragment $F_i$, the designated leader vertex $r_i$ of $F_i$ receives the aggregated outcome of convergecast function $f$ over all the inputs of vertices $v\in F_i$ (see Claim \hyperref[CB.2_ConvCastCorrectness]{B.2}). A convergecast communication cycle has time complexity of $O((d(T)+K)\log n)$ rounds (see Claim \hyperref[CB.3_ConvCastTimeCplx]{B.3}), message complexity of $O((d(T)\cdot K)\log n)$ (see Claim \hyperref[CB.4_ConvCastMsgCplx]{B.4}) messages, and memory complexity of $O(\log^2 n)$ bits (see Claim \hyperref[CB.5_ConvCastMemCplx]{B.5}).
%\subsection{Analysis}

\section{Performing a Broadcast through a Communication Cycle}
\label{4_Broadcast}
The broadcast cycle enables each fragment $F_i\in \mathcal{F}$ to transmit a message $msg_i$ from a leader vertex $r_i\in F_i$ to all vertices in $F_i$. We devise a method to send messages from the leader $r_i$ of fragment $F_i$ to all the vertices of $F_i$ using two steps, mirroring the convergecast cycle. First, we perform an extrinsic step, where we route the messages from each fragment leader $r_i$ to its respective subset roots. Next, we perform an intrinsic step, broadcasting the messages over the disjoint spanning trees of the subsets. In order to convey messages from fragment leaders to subset roots, we make use of routing labels. 

\subsection{Arrangement of Time Slots}
\label{4.1_slotOrdering}
A fragment $F_i$ can have up to $O(K_b)$ subsets, and its leader vertex $r_i$ needs to be able to broadcast a message to the leaders of all these subsets. Recall that $K_b = |\bigcup_{i=1}^K\mathcal{B}_i|$. Using just polylogarithmic memory per vertex, it is impossible to store the routing labels of all the subset roots in the memory of fragment leader $r_i$ simultaneously. Instead, we pair the outgoing messages with the routing labels of subset roots at the BFS tree root $r_T$. This is made possible by assigning additional time slots to subsets, and by arranging the time-slots of leader vertices and subsets in a particular order. 

We assign a time slot to each fragment and also to each subset. Hence, the overall number of required slots is $K+K_b\leq2K_b$. Suppose that a fragment $F_i$ consists of $k_i=|\mathcal{B}_i|$ subsets. We assign a sequence of $k_i+1$ consecutive slots to $F_i$,  $\{j_i,j_i+1,\ldots,j_i+k_i\}\subset\{1,\ldots,2K_b\}$. For every $i \in [K]$, denote by $j_i$ the slot associated with the fragment $F_i$. We give the first slot $j_i$ to the fragment as a whole, and store its value with the leader vertex $r_i$. The leader uses its slot to deliver the broadcast message to the BFS-root $r_T$. Next, subsets use their slots to send queries to $r_T$, requesting the message from their leader. These queries carry the respective routing labels of the roots of subsets. The BFS-root $r_T$ stores the message received from the leader of fragment $F_i$ for the following $k_i$ rounds, until a new message arrives from a different fragment leader.

We refer to the leader slot $j_i$, used for publishing the message $msg_i$, as a \textbf{publishing slot} or a \textbf{p-slot}. We refer to the next slots $\{j_i+1,\ldots,j_i+k_i\}$, used by subset roots to send queries, as \textbf{query slots} or \textbf{q-slots}. We call the whole range $\{j_i,j_i+1,\ldots,j_i+k_i\}$ a \textbf{p-q range}. The set of ranges for all the fragments $\mathcal{F}=\{F_1,\ldots,F_K\}$ in the graph $G$ constitutes a \textbf{p-q slot set}. In this section we will assume the slots are in place, correctly assigned to fragments and subsets in non-overlapping consecutive ranges. In Section \hyperref[5.3_SlotSetGen]{5.3} we describe a procedure that calculates and assigns slot sets at every phase of the algorithm, and show how this procedure is integrated in the MST construction.

\subsection{The P-Q Broadcast Cycle}
\label{4.2_pq_broadcast}
Next, we describe a broadcast cycle. First comes the extrinsic step, mirroring the extrinsic step of the convergecast cycle. The extrinsic step takes place in a pipelined fashion. As in the convergecast cycle in Section \hyperref[3.2_description]{3.2}, every time slot carries an upcast over BFS tree $T$, towards the root $r_T$. Slot uniqueness guarantees that messages originating in distinct slots are using distinct depths of $T$ in each round (see Claim \hyperref[CB.1_CycleSeparatesFrags]{B.1} for detailed proof). It follows that $r_T$ receives one message in each round $t=i+d(T)$, for $i\in[1,K+K_b]$. 

Every slot $j$ carries a message to $r_T$ in round $j+d(T)$. The message consists of two elements, type and content. The first element is a bit indicating the message type, either publishing (P) or query (Q). The second element varies according to type. For a message of type P sent by the leader vertex of $F_i$, the content is the message $msg_i$ that $F_i$ is broadcasting. For a message of type Q sent by a subset $B$, the content is the routing label $lbl_B$ of the root vertex $r_B$ of the subset $B$. 

The BFS-root $r_T$ handles a message of type P $\langle P, msg_i \rangle$ by storing $msg_i$ in a local variable $msg_T$, overwriting any message that was stored there previously. The root $r_T$ handles each query message $\langle Q, lbl_B \rangle$ by composing a response $\langle msg_T, lbl_B \rangle$ using the locally stored message, and sending it back down the BFS tree $T$, to be routed according to the label $lbl_B$. 

For every query from a subset root, there arrives a response within $2d(T)$ rounds of sending the query. Once all responses have reached their destination subset roots, in $2d(T) + K + K_b$ rounds, the extrinsic step is completed. Next, subsets carry out the intrinsic step. Subset roots initiate parallel broadcasts of their received responses over the spanning trees of their respective subsets.

When a broadcast cycle is completed, for each fragment $F_i$, if designated leader vertex $r_i$ sent a message then all vertices $v\in F_i$ receive it (see Claim \hyperref[CB.6_BCastCorrectness]{B.6}). A broadcast communication cycle has time complexity of $O((d(T)+K)\log n + D_b)$ rounds (see Claim \hyperref[CB.7_BCastCTimeCplx]{B.7}), message complexity of $O((d(T)\cdot K)\log n + n)$ (see Claim \hyperref[CB.8_BCastCMsgCplx]{B.8}) messages, and memory complexity of $O(\log^2 n)$ bits (see Claim \hyperref[CB.9_BCastCMemoryCplx]{B.9}).
%\subsection{Analysis of the Broadcast Cycle}

\section{Constructing MST via Communication Cycles }
\label{5_MSTConstruction}
Our algorithm for constructing MST, like the algorithms of \cite{gkp, KP1998fast,PRS17, E20Simple, haeupler2016low, haeupler2018round}, consists of two parts. The first part conducts Lenzen's variant of GKP (see Section \hyperref[2.3.2_LenzenGKP]{2.3.2}). In the first part, communication within and between fragments is carried out over the MST itself. The second part conducts basic GKP phases (see Section \hyperref[2.3.1_GKP]{2.3.1}): they proceed by finding MWOEs, partitioning the emergent virtual trees into subtrees with bounded depth, and merging the subtrees. 

In the second part of the algorithm fragments perform broadcast and convergecast through communication cycles. In Section \hyperref[5.1_GlobalSetup]{5.1} we describe the setup procedure that takes place once, just before the second half of the algorithm begins. The setup enables communication cycles in the first phase. In Section \hyperref[5.2_GKPPhaseOverCycle]{5.2} we review the GKP phase, and show that the communication cycle supports all the types of communication necessary to carry out a GKP phase. In Section \hyperref[5.3_SlotSetGen]{5.3} we describe a memory-efficient method for recalculating the slot set after each phase. In Section \hyperref[5.4_RefToPWA]{5.4} we describe the connection of the MSF problem, and of our solution, to the PWA problem. In Section \hyperref[5.5_Cplx]{5.5} we summarize the overall complexity of the MST construction procedure.

\subsection{Meeting the Prerequisites for Communication Cycles} 
\label{5.1_GlobalSetup}	

In order to meet the prerequisites of the communication cycle, a BFS tree $T$ spanning the entire graph $G$ is constructed, rooted at some vertex $r_T$. Next, routing tables and labels for tree routing over $T$ are computed. Every vertex $v\in V$ is assigned an $O(\log^2 n)$-bit label $lbl_v$, as well as an $O(\log n)$-bit routing table that will enable $v$ to route messages over $T$ using these labels. Here we use the compact tree routing scheme devised by \cite{TZ01Compact}, described in Appendix \hyperref[A.5.2_TZ_Tree_Routing]{A.5.2}.

%<redundant with intro>Communication cycles require fragments to be partitioned into disjoint, connected subsets (see Section \hyperref[3.1_preq]{3.1}). In each GKP phase, we will distinguish between fragments and \textbf{base fragments}. Similarly to \cite{E20Simple, PRS17}, we will use the term ``base fragments'' to refer to the MST fragments discovered by the end of the first part of the algorithm.

In each phase, fragments act as virtual nodes which connect to form trees and then partition into bounded-depth subtrees. To do all this, fragments use the communication cycle mechanism. In the first phase of the second part of our algorithm, every fragment consists of a single base fragment. In subsequent phases, every fragment will be the result of merging the subtrees of fragments from the previous phase. Thus, in subsequent phases every fragment will inherit a partition into base fragments from its components.

In order to proceed with cycle communication, every vertex $v\in V$ needs to know the slot numbers of its fragment and of its subset, as well the routing label of the leader vertex of its fragment. 
In the first phase of the second part, a single slot number for each base fragment is sufficient. Slot numbers meet requirements as long as they are all unique values in the range $[1,\ldots,K_b]$. To assign slot numbers, we use interval allocation \cite{santoro1985labelling} over the auxiliary BFS tree $T$ (see Fig. \hyperref[Fig2_rangePropagation]{2}).  Every root of a base fragment requests and receives a distinct interval of size $1$ from the range $[1,\ldots,K_b]$. Next, these roots send a broadcast over the base fragment trees, notifying other vertices in the base fragment about their slots. Every root vertex of a base fragment also becomes a fragment leader, and announces its routing label in the same way.

\begin{figure}[h]
	\begin{center}
		
		\label{Fig2_rangePropagation}
		
		\includegraphics[width=0.7\textwidth]{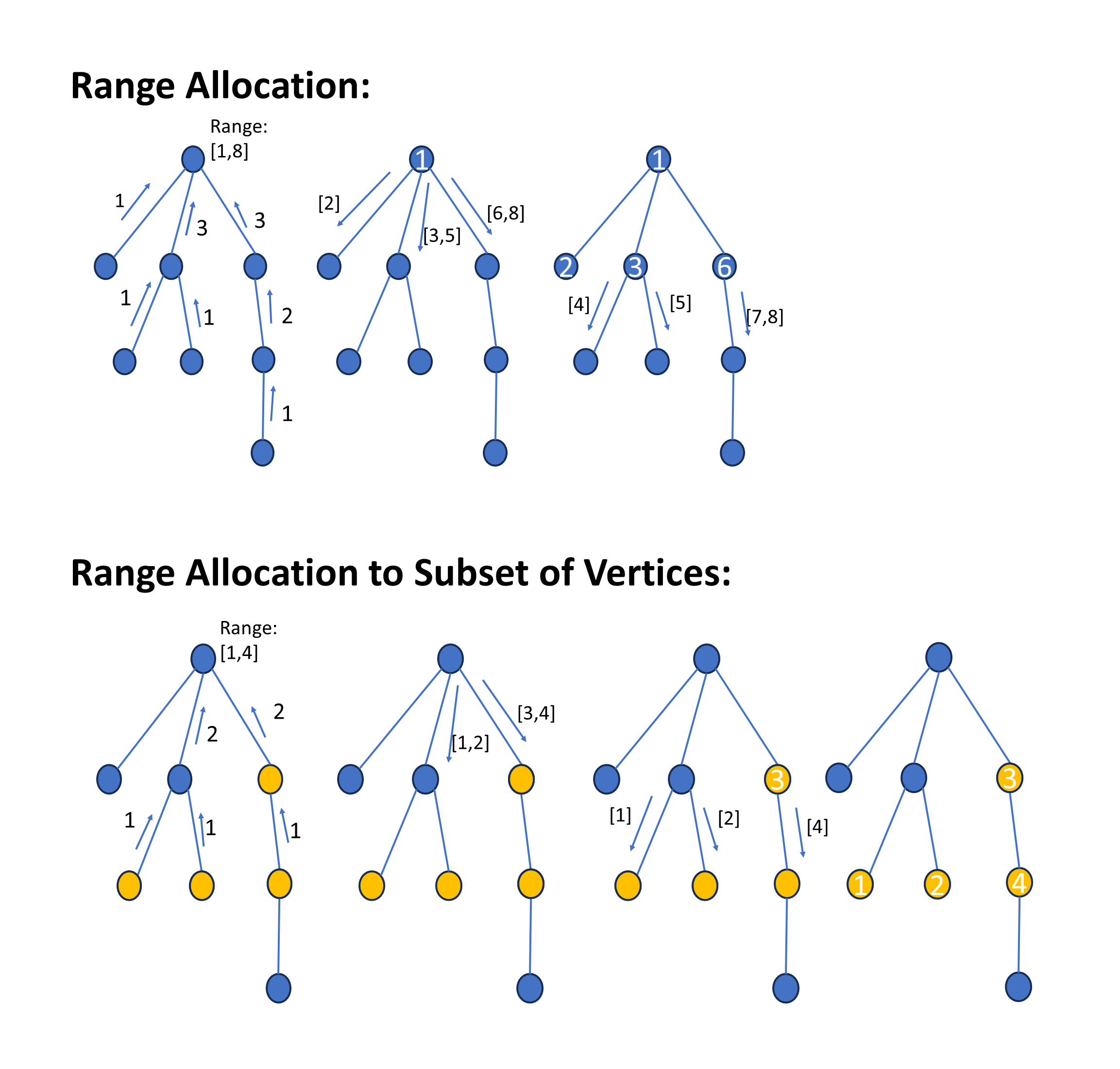}
		
		\caption{The top diagram illustrates an allocation of one slot to every node. The bottom diagram illustrates an allocation of slots only to the yellow nodes.}
		
		\Description{In the top row, three diagrams showing a tree. In the first diagram, arrows pointing up the tree along the edges, labeled with the number of descendants that each vertex has in its subtree (i.e. 1 from every leaf node, 3 from a node with two children, etc.). In the next two diagrams, arrows pointing down along tree edges - depth 0 to 1 in the second diagram and depth 1 to 2 in the third diagram. Arrows are labeled with ranges of numbers. Range sizes match descendant numbers from the first diagram. Vertices above the range-labeled arrows are marked with a number taken from the range. In the bottom row, four diagrams showing a tree which has some nodes marked yellow. In the first diagram, arrows pointing along every edge, labeled with the number of descendants that are marked yellow. In the next diagrams, arrows pointing down along tree edges, as in the top row. Yellow-colored vertices above the range-labeled arrows are marked with a number taken from the range.}
		
	\end{center}
\end{figure}

The following claims establish the correctness and complexity of the setup procedure.

\begin{Claim}
\label{C5.1}
After performing the following actions: Constructing a spanning BFS tree $T$, calculating routing tables and vertex labels for routing over $T$, assigning a distinct time slot to each base fragment; we can use communication cycles in the first phase of the algorithm that computes MST.
\end{Claim}

\begin{proof}
	\label{C5.1_Proof}
	We review the prerequisites described in Section \hyperref[3.1_preq]{3.1}. 
	
	The prerequisites are: 
	\begin{enumerate}
		\item There is a tree $T$, spanning the graph $G$, and every vertex $v\in V$ knows the overall depth of the tree $d(T)$ and its own depth $d(v)$ in $T$.
		\item A tree-routing scheme over $T$ is in place. Every vertex $v$ knows its own routing label. Every vertex $v_b$ in a base fragment $B$ knows the routing label of $B$'s fragment leader vertex.
		\item Every fragment has a unique slot number from the range $[1,\ldots,K_b]$.
		
		\item Every base fragment $B$ has a spanning tree $T_B$. There is an upper bound on both the number of subsets $K_b$ and the depth $D_b$ of the spanning trees which span the subsets.
	\end{enumerate}
	
	Prerequisites 1-3 can be satisfied by performing the procedures listed above: BFS tree construction, routing scheme setup, interval allocation. Prerequisite 4 is satisfied by the base fragments that exist after the completion of the first part of our algorithm, with $D_b$ and $K_b$ given below. After $O(\log n)$ Lenzen phases it is possible to create base fragments that satisfy the requirement that their overall number is bounded by $K_b=O(\min{\{\sfrac{n}{D},\sqrt{n}\}})$, and the diameter of the MST-subtrees spanning them is bounded by $D_b=O(\max{\{D,\sqrt n\}})$ (see \cite{Lenzen2016}, Ch. 6 or \cite{E20Simple}, Section 4). 
	
	The arrangement of time slots (see Section \hyperref[4.1_slotOrdering]{4.1}) imposes the additional requirement that subsets should have unique slots too, and that they consecutively follow fragment slots. This trivially holds for the first phase, because every fragment consists of one single subset.
	
\end{proof}

\begin{Claim}
\label{C5.2}
	The setup procedure has time complexity of  $O((D+\sqrt{n})\log n)$ rounds, message complexity of $O(m+ n\log n)$ messages, and memory complexity of $O(\log^2 n)$ bits.
\end{Claim}

\begin{proof}
	\label{C5.2_Proof}
	Time requirements consist of $O(D)$ rounds for the construction of the BFS tree, $O(D\log n)$ rounds for the setup of the routing tree \cite{TZ01Compact}, and $O((D+\sqrt{n})\log n)$ rounds for broadcasts of leader labels from the base fragment roots. The time complexity of the setup process is thus, in total, $O((D+\sqrt{n})\log n)$ rounds. Message complexity for the construction of the BFS tree is $O(m)$ (\cite{peleg2000distr}, Ch. 5). The other steps are broadcasts of $O(\log^2 n)$-bit messages, over a tree spanning $G$ or over disjoint subtrees of such a tree, and so their message complexity is $O(n\log n)$. 
	
	The memory complexity of the steps of the setup procedure - constructing a BFS tree $T$, setup for tree routing, interval allocation - does not exceed the memory complexity of the information that they generate. Memory is required for the storage of the tree $T$ itself, the routing labels and tables, and the time slots. Specifically, each vertex $v$ is required to store $v$'s parent edge and number of children in the tree $T$. Additionally, each vertex $v$ needs to store $v$'s routing label $lbl_v$, and the routing label $lbl_{B}$ of the leader vertex $r_B$ of $v$'s base fragment $B$. Finally, each vertex $v$ needs to store the value of a single slot number. All these variables have memory complexity of $O(\log n)$ bits, except for the vertex routing labels, which require $O(\log^2 n)$ bits.
\end{proof}

After every phase, compatibility with prerequisites needs to be reestablished. Every phase $t>1$ requires a new set of slot numbers, which meets the requirements described in Section \hyperref[4.1_slotOrdering]{4.1} (see also Claim \hyperref[CB.6_BCastCorrectness]{B.6}). For each phase $t>1$, slot numbers are computed in the final step of the preceding phase. We discuss this further in Section \hyperref[5.3_SlotSetGen]{5.3}. 

%\label{Appendix.F.2_BeforeBfrags}
When we apply our algorithm to solving the general MSF problem, we are required to handle an additional case in the first part of the algorithm. Some fragments may stop growing before reaching the size and diameter threshold of base fragments. We call these fragments \textbf{small components}. We stipulate that, up to the moment of transition to communication cycles, any fragment that fails an attempt to find MWOE candidates belonging to $E_H$ will terminate. At the time of transition, every terminated fragment is designated as a small component. Such a fragment does not go on to become a base fragment. This guarantees that:

\begin{enumerate}
	\item Base fragment size is subject to the same lower bound as in the MST case, and so the upper bound on the number of base fragments still holds.
	\item Small components are of size $O(D+\sqrt{n})$, and so can carry out broadcast and convergecast computations over their respective minimum spanning trees in $O(D+\sqrt{n})$ rounds, without using communication cycles.
\end{enumerate}

A small component $S$ does not request slots during setup, or at any later point. The root vertex $r_S$ of a small component $S$ is not required to publish its routing label to the vertices in $S$. During a communication cycle, a vertex $v\in S$ is required to aggregate and upcast pipelined messages from other fragments, and to route messages, but never sends any new messages on behalf of $S$.

\subsection{Carrying out a Single Phase Using the Cycle Protocol}
\label{5.2_GKPPhaseOverCycle}
The GKP phase \cite{gkp} proceeds in three stages. First, in order to discover the MWOE, each fragment carries out a convergecast calculation. Secondly, in order to partition the trees of the emerging virtual forest into $O(1)$-depth subtrees, fragments communicate with their immediate neighbors on these virtual trees. Finally, in order to merge the subtrees and create the fragments of the next phase, a broadcast over these virtual subtrees is carried out.

We will show that the GKP phase can be executed, with polylogarithmic memory complexity, using communication cycles. The first step, finding MWOEs, forms a virtual graph - a forest of virtual trees. Roughly speaking, the rest of the GKP phase executes distributed algorithms on these trees. Communication cycles help fragments to emulate vertices in the virtual tree. 

Communication cycles enable fragments to efficiently coordinate the states of virtual vertices. Once every vertex $v$ in a fragment $F$ has a copy of the state of $F$, every vertex $v$ is able to compose the same message on behalf of $F$. In particular, a vertex $v$ that is adjacent to an MWOE $e$ is able to compose the correct message. Vertex $v$ is then able to convey the message to the MWOE-adjacent fragment, in $O(\log n)$ rounds. Fast coordination is possible because the distributed algorithms that we use are memory-efficient, and the state of a virtual vertex can be represented using $O(\log^2 n)$ bits.

\subsubsection{Discovery of MWOEs by Fragments}
\label{5.2.1_MWOEDiscovery}
Every fragment $F$ may consist of multiple base fragments which function as subsets, and communication within $F$ is conducted via convergecast and broadcast cycles. A single convergecast cycle is required for the leader vertex $v_F$ of a fragment $F$ to discover the identity of the MWOE $e=\{v_{adj},u\}$, which connects $F$ to some neighboring fragment $F'$. Afterwards, a broadcast cycle notifies every vertex in the fragment $F$ about the discovered MWOE. Specifically, the vertex $v_{adj}\in F$ adjacent to $e$ is notified. At this point the edge $e$ is recorded by $v_{adj}$ as an MST edge.

Notifying $v_{adj}$ enables communication between fragments $F$, $F'$ in the next stages of the phase. The MWOE notification to $v_{adj}$ includes the ID of $F'$, which at this point becomes the parent of $F$ in a virtual fragment tree. The notification also enables the vertex $v_{adj}$ to identify the MWOE edge $e$. Later, whenever $F$ will decide to message its parent $F'$, the vertex $v_{adj}$ will deduce that it is responsible for sending the message. %At this point the vertex $v_{adj}$ will also deduce the edge $e$ that connects it to $F'$.

When we apply our algorithm to solving the general MSF problem, we are required to handle an additional eventuality during MWOE discovery. If a fragment $F$ reaches its maximal size and spans its connected component at the end of phase $t-1$, in phase $t$ it will be unable to discover MWOE candidates from $E_H$. In phase $t$, the fragment becomes a \textbf{terminated fragment}. MWOE discovery during phase $t$ is the point in the algorithm's execution where fragment $F$ discovers that it is terminated, because MWOE candidates fail to arrive. In this case, fragment leader $v_F$ uses the broadcast cycle to inform all vertices of $F$ that the fragment is terminated. After this broadcast cycle, we add an additional convergecast cycle over BFS in order to detect non-terminated (i.e., active) fragments. Fragment roots announce their status (terminated or active). The BFS root $r_T$ determines when no active fragments remain, and terminates execution.

\subsubsection{Communication with Adjacent Fragments}
\label{5.2.2_CommWithAdjFragments}
In the following stages of the GKP phase, fragment trees find maximal matching \cite{gkp}, and subdivide into subtrees around matched pairs. Next, we argue that all necessary communication can be carried out efficiently using broadcast and convergecast cycles within fragments, and single-round messages over the MWOEs of the ongoing phase. In Section \hyperref[5.3_SlotSetGen]{5.3}, we show that communication cycles are sufficient to support the process of slot set generation as well.

GKP phases can be executed using three types of message exchanges (listed in Claim \hyperref[C2.1_GKP_MsgTypes]{2.1}) over the virtual trees. We show that these message exchanges can be carried out between fragments using communication cycles.

\begin{Claim}
	\label{C5.3}
	A virtual fragment tree can simulate any of the following three interactions using $O(1)$ communication cycles and $O(1)$ additional rounds of messaging over MWOEs.
	\begin{enumerate}	
		\item A round of messaging from each parent to all its respective children.
		\item A round of messaging from each parent to one arbitrarily selected child.
		\item A round of messaging from all children to their respective parents, followed by aggregation with a semigroup function of the received messages.
	\end{enumerate}
	
\end{Claim}

\begin{proof}
	\label{C5.3_propf}
	Suppose that a sender fragment has an $O(\log n)$-bit message to send in one of these three ways. We review the required steps.
	
	First, the sender fragment identifies the edges over which the message needs to be sent. In the first case, every child pings its parent, and receives the parent's message in response. In the third case, the MWOE-adjacent vertex on the child side already knows the edge. In the second case (where every parent sends a message to a single arbitrary child), each parent fragment receives multiple queries. Then, each parent fragment participates in a convergecast cycle to elect a single query, and in a broadcast cycle to announce the elected query to the correct MWOE-adjacent vertex.
	
	Next, recipient fragments aggregate received messages and announce the outcome. This is done with a convergecast cycle and a broadcast cycle. In the first and second cases every child fragment contains a single vertex that knows the message, and a convergecast cycle brings these messages to their respective fragment leaders. In the third case, recipient parent fragments participate in a convergecast cycle to aggregate received values. Once the leaders of recipient fragments know the respective outcome values, a broadcast cycle is used to announce them. 
\end{proof}

Table \hyperref[Tbl2_VirtualRoundStagesVSMSGTypes]{2} describes these three procedures, along with their respective complexities.

\begin{table*}[h]
	\label{Tbl2_VirtualRoundStagesVSMSGTypes}
	\centering

	\caption{Virtual Round Complexity, By Message Type}
	\begin{tabular}{p{0.15\textwidth} p{0.25\textwidth} p{0.25\textwidth} p{0.25\textwidth}}
		\toprule
		Stage & Parent to all Children & Parent to Arbitrary Child& Children to Parent \\ [0.5ex] 
		\midrule
		Source fragments select edge & Queries from children& Queries from children, convergecast and broadcast cycles to select and announce a child& MWOE is known on child side\\ 
		\midrule
		Messages sent & Response to query & Elected vertex responds to query & Messages over MWOE \\ 
		\midrule
		Received messages aggregated & Convergecast conveys message to the root of each fragment & Convergecast conveys message to the root of the elected child fragment & Convergecast aggregates multiple received messages to the root of the parent\\ 
		\midrule
		Aggregate message published & Broadcast outcome from the root of each child& Broadcast outcome from the root of the elected child& Broadcast outcome from the root of the parent\\
		\midrule
		Overall complexity & 2 rounds, 2 cycles & 2 rounds, 4 cycles & 1 round, 2 cycles  \\
		\bottomrule
	\end{tabular}
\end{table*}

\subsubsection{Maximal Matching and Tree Partitioning}
\label{5.2.3_MaxMatching}
In the GKP phase (see Section \hyperref[2.3.1_GKP]{2.3.1}), the virtual fragment trees partition into subtrees of depth $O(1)$. Specifically, subtree depth is at most 3. We recall the subroutines used in this step: 3-coloring \cite{CV}, maximal matching based on the coloring, and attaching unmatched nodes to a parent. By Claim \hyperref[C2.1_GKP_MsgTypes]{2.1}, each step consists of messages to all children, messages to an arbitrary child, and messages from children to their respective parents. We further observe that messages to parents are used either for counting children, or for notifying parents that children exist. Thus, messages from children to parents can be aggregated through summation and logical AND/OR.

It follows that the GKP phase can be executed using only the three types of messages between fragments that are summarized in Table \hyperref[Tbl2_VirtualRoundStagesVSMSGTypes]{2}. Since the GKP phase requires $O(\log^* n)$ rounds, carrying it out with virtual nodes simulated by fragments will require $O(\log^* n)$ communication cycles.

\begin{corollary}
		\label{C_5.4}
	A GKP phase can be carried out using $O(\log^* n)$ virtual rounds, i.e., its overall time, message and memory complexity is $O(\log^* n)$ communication cycles.
\end{corollary}

Recall also that at the start of each phase, there is a single round in which  vertices discover their respective MWOE-candidates. This round requires $O(m)$ messages. This round also has memory complexity $O(\log n)$. 

If we are solving the MSF problem, some fragments may already be terminated during this step. Terminated fragments cannot be a part of a tree with more than one fragment. Therefore terminated fragments do not take part in tree partitioning. A terminated fragment, since it has no neighbors, is not required to participate in maximal matching, and doesn't need to send messages or to handle incoming messages. Vertices contained in terminated fragments, similarly to vertices in small components, take part in communication cycles by forwarding traffic. However, vertices in terminated fragments do not inject input on behalf of their fragments.

%\subsubsection{Merging the Fragment Subtrees}
%\label{5.2.4_MergingFragTrees}
Finally, to complete a phase fragment subtrees are required to merge into single fragments. When fragments merge in the original GKP phase, a broadcast of the fragment ID is sent over each subtree, from the root to the other nodes. To sustain communication cycles, a more complicated process needs to take place, carrying more information. This process follows paths of broadcast and convergecast over the subtrees created in this section.

%In order to become a single fragment in the next phase, a set of base fragments requires a single leader and a consecutive slot range. To prepare for the next phase, first, every subtree calculates the number of base fragments that it contains. Then, a p-q slot set is generated so that it contains a distinct slot range for every subtree, with the correct number of slots. Finally, slot numbers are delivered to their respective base fragments and fragment leaders. 

%We explain this procedure in detail in Appendix \hyperref[Appendix.C_SlotSetGen]{C}. Roughly speaking, first every subtree holds a convergecast to count the base fragments contained in its fragments. Then, subtrees reserve distinct intervals from the new slot range (using the interval allocation procedure from Appendix \hyperref[5.3.1_subrangeAlloc_subroutine]{C.1}). Finally, subtrees propagate the slots from their respective intervals to their fragments and then to their base fragments, again using the interval allocation procedure. Further details can be found in Appendix \hyperref[5.3.2_slotSetGen]{C.2}. 

%The following claims (proved in appendix \hyperref[Appendix.C_SlotSetGen]{C}) establish the correctness and complexity of this computation:

\subsection{Generating the Slot Set After a Phase}
\label{5.3_SlotSetGen}

%\label{5.3_CreatingSlotSet}
After the partition of virtual trees into subtrees is completed, the next step is to construct a new slot set. In the next phase, the new slot set will enable subtrees to become fragments, and to use cycle communication as fragments. Each newly formed fragment requires a distinct, consecutive p-q range (see Section \hyperref[4.1_slotOrdering]{4.1}).

To fulfill the requirements of the broadcast communication cycle, in every fragment $F$, the leader vertex $r_F$ needs a slot range of a size matching the number of subsets (i.e., base fragments) in $F$. If the fragment $F$ contains $k$ subsets, the leader $r_F$ will require $k$ slots - a slot for every subset, and a single slot for the leader itself. The leader of a fragment is also the root of some base fragment, and that base fragment will not require a slot, so we have $k$ slots required in total.

The process of subrange distribution is similar to the allocation of intervals for interval routing \cite{santoro1985labelling}. There, every node that requires an interval makes a request to its parent, and requests are aggregated in a convergecast until they reach the root. Each request is met with a distinct interval in response. In this way, sub-intervals propagate down the tree (see Fig. \hyperref[Fig2_rangePropagation]{2}). Here the process of aggregating requests, and retracing their path back to origin, is carried out across two levels of hierarchy. Requests flow from base fragments to fragments, and then from leaf fragments to root fragments in virtual fragment subtrees.

\subsubsection{The Interval Allocation Subroutine}
\label{5.3.1_subrangeAlloc_subroutine}
The \textit{interval allocation subroutine} is used several times in the process of generating slot sets. This subroutine requires one broadcast communication cycle. By participating in the subroutine, a fragment leader vertex that has an interval $[N,N+S)$ can distribute it among its base fragments. When several fragment leaders have intervals, they can all participate in the subroutine within the same cycle.

Suppose that at the start of the cycle, a leader vertex has a number $N$ representing the bottom of its interval. Suppose that each base fragment $B$ requires an interval of size $S_B$. In the ultimate slot allocation, every base fragment requires a single slot. However, over the course of the slot set generation process, base fragments submit requests for slot intervals of various sizes, in order to propagate those intervals to adjacent fragments.
\begin{itemize}
	\item Every participating leader vertex uses its p-slot to upcast $N$ to the root $r_T$ of the BFS tree $T$.
	
	\item Every participating base fragment root $r_B$ upcasts its desired interval size $S_B$. Once the base fragment receives, in response, some value $N_0$, it claims the interval $[N_0,N_0+S_B)$.
	
	\item The root $r_T$ acts as an arbiter for all participating fragments by executing Algorithm \hyperref[Alg2_SubrangeAlloc]{2}.
\end{itemize}

\begin{algorithm}
	\label{Alg2_SubrangeAlloc}
	\caption{Interval Allocation}
	\begin{algorithmic}[1]
		%\Require BFS tree root $r_T$ Receives a message $\langle TYPE, VALUE, LABEL\rangle$
		\State Set local value $rangebottom := 1$
		
		\For{each round of the cycle}
		\If{tree root $r_T$ receives a message $\langle type, value, label\rangle$}			
		\If{Message type $type = P$}
		\State Set $rangebottom := value$.
		\ElsIf {Message type $type = Q$}
		\State Send message $\langle rangebottom, label\rangle$ over $T$, to be routed according to label.
		\State Update $rangebottom := rangebottom + value $
		\EndIf
		\EndIf		
		\EndFor
	\end{algorithmic}
\end{algorithm}

The following claim is immediate and we state it without proof:
\begin{Claim}
	\label{C5.5_IntervalAllocSubroutine}

	\begin{enumerate}
		\item Suppose that at the start of the cycle, a fragment $F$ has an interval $[N,N+S)$. Suppose also that the base fragments $B_1, B_2, \ldots B_k$ that make up $F$ have range requests with sizes $S_1, S_2, \ldots, S_k$ so that $\sum_{i=1}^{i=k} S_i \leq S$. Then the intervals claimed by the base fragments at the end of the cycle will be disjoint, and contained in the interval $[N,N+S)$.
		\item If the intervals held by several fragments $F_1, F_2, \ldots$ at the start of the interval allocation cycle are disjoint, and if request sizes match interval size for each participating fragment, then the intervals held by base fragments at the end of the cycle will also be disjoint.
	\end{enumerate}
	
\end{Claim}

We take note of a variant of the interval allocation procedure where there is no fragment leader upcasting a range bottom, and the range bottom is effectively $1$. In that case, if some vertices $\{v_1, v_2, \ldots\}\in V$ upcast requests with interval sizes $S_1, S_2, \ldots, S_k$, they will receive disjoint intervals $I_1, I_2, \ldots, I_k$, so that their union is the interval $[1, \sum_{i=1}^{i=k}S_i)$. This variant will be used in the procedure that generates slot sets. It will be used in order to request disjoint intervals for the next phase's fragments, before each of these new fragments begins to divide its interval among its components.

\subsubsection{Generating the Slot Sets}
\label{5.3.2_slotSetGen}

%Recall that after a GKP phase is complete, fragments are organized in subtrees of depth $O(1)$. Each vertex in these subtrees is a fragment. Each fragment may consist of many base fragments (or, equivalently, subsets). Our goal is to turn every subtree into a single fragment. For the leader vertex of this single fragment, we designate the vertex which served as leader of the fragment at the root of the subtree. 

At the beginning of the procedure that generates slot sets, in each subtree, the designated leader vertex collects slot requests from subsets, from all the fragments in the subtree. As a result, the designated leader vertex discovers the number of subsets in its subtree. Then, each leader requests and receives from the root $r_T$ of BFS tree $T$ a distinct interval from the range $[1, 2K_b]$. Finally, intervals are divided and propagated back to leader vertices of fragments that made requests. Fragment leader vertices then distribute slots from these ranges to their base fragments.

A challenge arises when requests for slots need to be delivered from subsets to leader vertices, and to the leader vertices of the respective root fragments of their subtrees. The challenge lies in delivering requests through all these steps in a way that enables each request to be traced back later to its origin. 

The requests travel through a path similar to the path taken by a child-to-parent message, as described in Claim \hyperref[C5.3]{5.3} (see also Table \hyperref[Tbl2_VirtualRoundStagesVSMSGTypes]{2}). Some steps are intrinsic convergecast steps - direct convergecast computations over base fragments. Other steps are extrinsic convergecast steps - pipelined convergecast computations over the BFS tree $T$. In order to enable tracing back requests, vertices at some points along the path cache the requests they send.

We describe the process below, and follow in detail the path taken by requests and partially aggregated requests. We specify the vertices that cache the partial sums that they upcast while participating in data transmission. We show that the cached information is sufficient to retrace the paths that were taken. We also show that the cached information takes up only $O(\log n)$ bits per vertex of additional memory.

The process consists of the following steps:

\begin{enumerate}
	\item A round of communication takes place, during which all child fragments ping their parent fragments. Parent fragments then participate in a convergecast cycle to count their child fragments in the virtual subtrees, formed in the preceding GKP phase.
	
	\item The following steps are iterated $d$ times, where $d$ is the upper bound on the depth of fragment subtrees (i.e., $d=3$). In each iteration, requests for intervals are sent from child fragments to parent fragments (see Fig. \hyperref[Fig3_rangeRequestPathIteration]{3} for an illustration).
	\begin{enumerate}
		\item Fragments that are ready send messages to their parent fragments, over the MWOE. A fragment $F$ becomes ready when it has received interval requests from all its child fragments in the previous iteration. Leaf fragments are considered ready in the first iteration. The message that fragment $F$ sends contains the size of the interval that $F$ is requesting. This size equals the number of subsets in all fragments downstream of the $F$, including $F$ itself. The MWOE-adjacent vertex $u_{adj}\in F$ that sends the message also caches the requested interval size at this point.
		
		\item Base fragments carry out the intrinsic step of a convergecast cycle. Any vertex $v_{adj}$ that received an interval request over an MWOE in (a) submits this request to the convergecast. Two convergecast calculations take place in parallel - summing up the interval sizes that were requested by child fragments, and counting the received messages. Every vertex that takes part in the convergecast caches the partial sum of interval request sizes that it upcasts.\footnote{Note that a base fragment $B$ may be required to perform convergecast in multiple iterations of this step. In that case, the cached interval sizes are updated in every iteration. Consider a single base fragment $B$ of fragment $F$ containing two vertices $v_1, v_2$ that are adjacent to MWOEs from two of $F$'s child fragments - $F_1$ and $F_2$, respectively. Suppose that $F_1$ sends a request for interval size $n_1$ in iteration $1$, whereas $F_2$ becomes ready and sends a request for size $n_2$ in iteration $2$. Suppose vertex $v\in B$ is a common ancestor of $v_1$ and $v_2$. Vertex $v$ will participate in the convergecast during both iterations. In iteration $1$, vertex $v$ will upcast $n_1$, and in iteration $2$ it will upcast $n_2$. In that case, vertex $v$ will cache $n_1$ in iteration $1$, and in iteration $2$ it will update the cached value to $n_1 + n_2$.}

		\item Every base fragment $B$ that was active in 2(b) participates in the extrinsic step of a convergecast cycle. The root vertex $r_B$ of $B$ upcasts the values it has obtained in step 2(b). These values are aggregated and routed to the leader vertex $r_F$ of $B$'s fragment $F$. Base fragment $B$ submits to the convergecast both its requested interval size, and its number of received messages. The root of every base fragment $B$ caches the size value that it sends. Every fragment leader vertex $r_F \in F$ maintains and updates a sum of received interval sizes. Leader vertex $r_F$ also counts the messages that it has received from child fragments since the start of step 2. If messages have arrived from every child of $F$, the fragment $F$ will become ready in the next iteration. 
		
		\item Each fragment $F$, if it is ready and is not the root of its subtree, participates in a broadcast cycle. The leader vertex $r_F$ of $F$ notifies the vertices of $F$ that $F$ is ready, and publishes the size of the interval that $F$ will request from its subtree parent in the next iteration. This size is the sum of interval requests from child fragments, and of the interval that $F$ itself requires.
	\end{enumerate}
	
	\item For fragments that are at the roots of their respective subtrees, the leader vertices are designated as leader vertices for the merged fragments of the next phase. Designated leaders initiate an interval allocation cycle. In this interval allocation cycle, there is no p-slot, and designated leader vertices use their current slots as q-slots. They request and receive intervals matching the number of base fragments in their respective subtrees. Upon receiving an interval, each designated leader vertex reserves the first slot from the interval as its p-slot.
	
	\item Every fragment $F$ that is a root of a subtree distributes the the interval that it obtained among its descendants in the virtual subtree. This step mirrors step 2, proceeding in the opposite direction. Like in step 2, it is iterated $d=3$ times (recall that subtrees have depth of at most 3).
	
	\begin {enumerate}
	\item If the leader vertex $r_F$ of a fragment $F$ has received an interval, in step 3 or in the previous iteration, $r_F$ reserves a sub-interval from the beginning for fragment $F$, of length equal to the number of subsets that $F$ contains.
	
	\item Next, $F$ participates in an interval allocation cycle, assigning intervals to base fragments according to requests. Base fragments re-send the requests for interval sizes that they cached in step 2(c). For each base fragment $B$, the requested size is the sum of requests from fragments that are adjacent to $B$ via incoming MWOE edges. Every fragment $B$ calculated this sum in step 2(b).
	
	\item For each base fragment $B$ whose root $r_B$ received an interval, a broadcast is held, retracing step 2(b). In step 2(b), over $B$ we summarized interval requests from adjacent child fragments of the fragment $F$ that contains $B$. Correspondingly, here we distribute the interval that was received by $r_B$ over $B$, among child fragments of $F$ that are adjacent to $B$ (via incoming MWOE edges). Every child vertex $v\in B$ re-upcasts the value that it cached in 2(b), and in response receives an interval. As a result, interval requests are traced back to their respective origins. Intervals are broken up into smaller parts where necessary, and propagated back to the MWOE-adjacent vertices that requested them.
	
	\item MWOE-adjacent vertices in child fragments re-send the requests they have cached in 2(a). MWOE-adjacent vertices in parent fragments, if they received an interval in step 4(c), reply to requests with sub-intervals. Next, receiving vertices in child fragments send the intervals to their respective fragment leaders, using a convergecast cycle.
	
\end{enumerate}

\item Fragment leader vertices distribute the slots from their respective intervals, reserved in 4(a), to their subsets. This is done using the interval allocation subroutine.

\end{enumerate}

\begin{figure*}[h]
\label{Fig3_rangeRequestPathIteration}
\begin{adjustwidth}{-9mm}{-9mm}
	\begin{minipage}[H]{0.5\textwidth}
		\centering
		\includegraphics[width=0.6\linewidth]{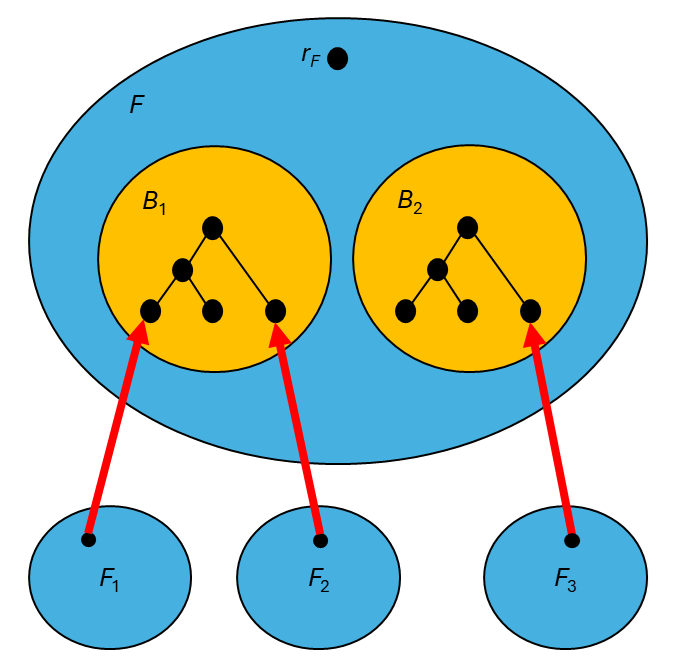}
		%\caption*{ZZZDiagram 1}
	\end{minipage}
	\hfill
	\begin{minipage}[T]{0.5\textwidth}
		\small
		\textbf{Step 2a.}  
		Child fragments $F_1$, $F_2$, $F_3$ send their interval request sizes to parent fragment $F$ over their respective MWOEs. The MWOE-adjacent vertices contained in $F_1$, $F_2$, $F_3$ cache the sizes that they send.
	\end{minipage}
	
	\begin{minipage}[H]{0.55\textwidth}
		\centering
		\includegraphics[width=0.6\linewidth]{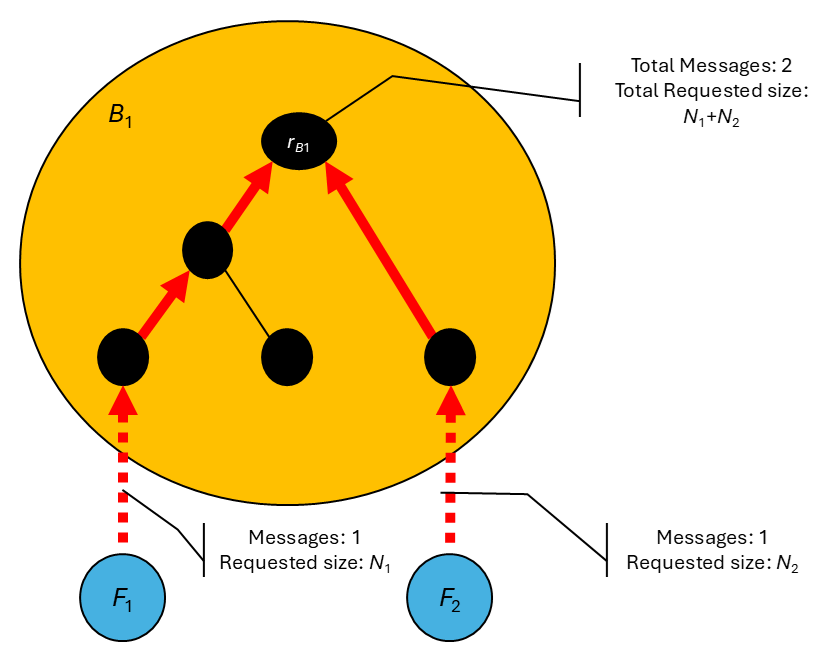}
		%\caption*{ZZZDiagram 1}
	\end{minipage}
	\hfill
	\begin{minipage}[T]{0.5\textwidth}
		\small
		\textbf{Step 2b.}  
		The intrinsic step of the convergecast cycle. A convergecast computation is held over base fragment $B_1$. The calculation summarizes the sizes that were requested by the fragments $F_1$, $F_2$ that are adjacent to vertices contained in $B_1$. Every vertex in $B_1$ caches the sum that it upcasts. The convergecast also summarizes the number of messages.
	\end{minipage}
	
	\begin{minipage}[H]{0.55\textwidth}
		\centering
		\includegraphics[width=0.6\linewidth]{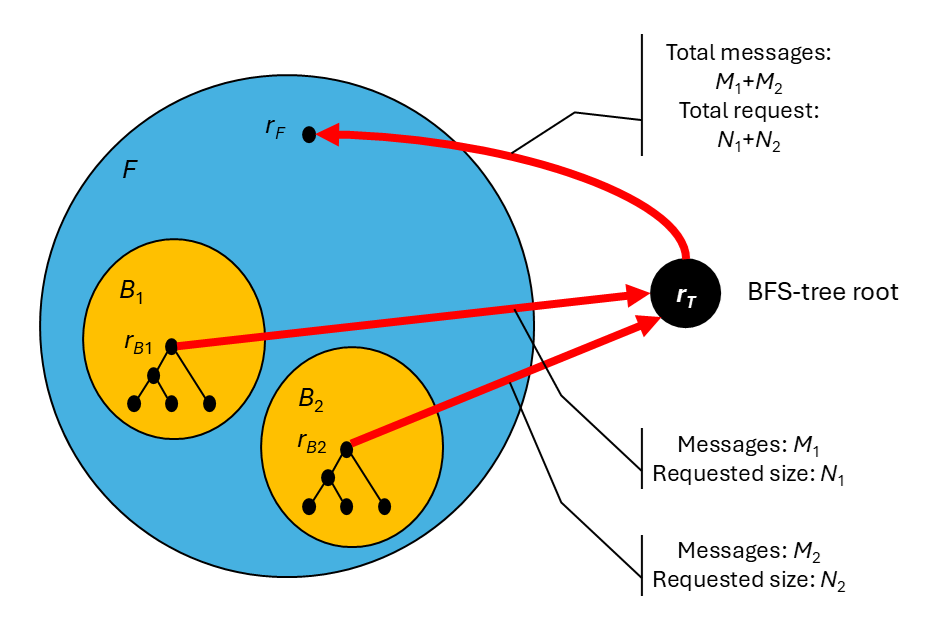}
	\end{minipage}
	\hfill
	\begin{minipage}[T]{0.5\textwidth}
		\small
		\textbf{Step 2c.}  
		The extrinsic step of the convergecast cycle. Pipelined convergecast computations, for each fragment $F$, are held over the BFS tree $T$. Base fragments $B_1$, $B_2$ contained in the fragment $F$ notify the leader vertex of $F$, $r_F$, about the number and total size of requests they have collected. Base fragment roots $r_{B1}$, $r_{B2}$ cache the sizes they request.
	\end{minipage}
	
	\begin{minipage}[T]{0.55\textwidth}
		\centering
		\includegraphics[width=0.6\linewidth]{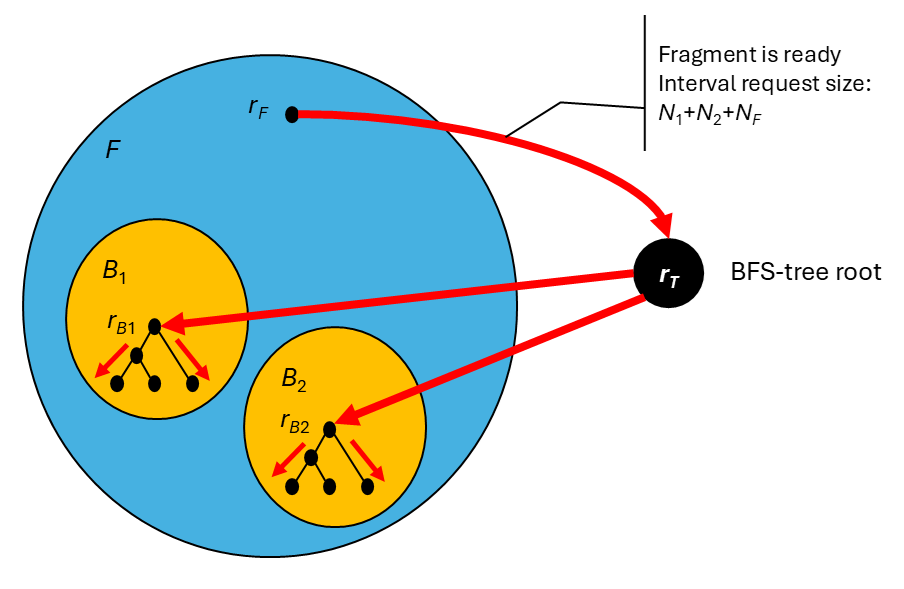}
		
	\end{minipage}
	\hfill
	\begin{minipage}[T]{0.5\textwidth}
		\small
		\textbf{Step 2d.}  
		A broadcast cycle. If messages from all children of $F$ have reached leader vertex $r_F$, then $r_F$ announces that $F$ is ready, and broadcasts the total interval request size of $F$. The size of the requested interval equals the sum of requests from children, and of the interval size $N_F$ needed for $F$'s base fragments.
	\end{minipage}
	\caption{A detailed illustration of a single iteration of step 2 of the slot set generation procedure described in Section \hyperref[5.3_SlotSetGen]{5.3}.}
\end{adjustwidth}
\end{figure*}

The following claims establish the correctness and complexity of the slot set generation procedure. 
\begin{Claim}
	\label{C5.6_SlotGen_correctness}
	For every phase, the slot set created using the above method is correct, i.e.: the leader vertex of every fragment holds the first slot, and the subrange for every fragment is consecutive, and does not intersect with other subranges. %See proof in Appendix \hyperref[Appendix.D_C5.5_proof]{C}.
\end{Claim}

\begin{proof}
	\label{C5.6_proof}
	By construction, in step 3 every designated leader vertex reserves the first slot of a distinct consecutive interval, and allocates the remainder of the interval to other fragments in the subtree. It remains to establish that the allocated intervals are large enough, and that slots reach the appropriate subsets. The correctness of interval sizes is guaranteed if the numbers of subsets which are calculated and submitted to leader vertices of root fragments in step 2 are correct. Slots reach the correct subsets if convergecast paths are retraced correctly in step 4.
	
	For every phase, we make the induction assumption that the current slot set is correct. We also assume that the last count of subsets was correct, and every fragment knows the number of subsets that it contains. These assumptions hold for the first phase, where every fragment consists of a single subset (i.e., base fragment). 
	
	Recall that step 2 is iterated three times. Step 2(a) is correct in the first iteration because fragments know their subset count. In every subsequent iteration, 2(a) succeeds if the previous iterations have been correct. Step 2(b) relies on convergecast and step 2(c) relies on the correctness of the convergecast cycle (See Claim \hyperref[CB.2_ConvCastCorrectness]{B.2}).
	
	We establish the correctness of step 4 from the correctness of the broadcast procedure used in 4(b) and the broadcast cycle used in 4(c). The broadcast cycle involves the interval allocation subroutine (see Claim \hyperref[C5.5_IntervalAllocSubroutine]{5.5}). The allocation of single slots in step 5 is also an instance of the interval allocation subroutine.
	
	We note for completeness that in order to satisfy all the requirements detailed in Sections \hyperref[3.1_preq]{3.1} and \hyperref[4.1_slotOrdering]{4.1}, every vertex is required to know the time slot and the routing label of its new fragment leader vertex. We observe that when the slot set creation process is complete, the leader vertex of every new fragment discovers its new time slot. Both the time slots and the routing labels of new leader vertices can then be announced by performing a broadcast over the virtual subtrees, using the old slot sets to communicate. This concludes the process of merging the new fragments.
\end{proof}

\begin{Claim}
	\label{C5.7_slotSet_Cplx}
	For every phase, the above method for creating the slot set has the time and message complexity of $O(1)$ communication cycles. The slot set created using the above method, and the creation process, both have memory complexity of $O(\log n)$. %See proof in Appendix \hyperref[Appendix.C_C5.6_proof]{C}.
\end{Claim}

\begin{proof}
	\label{C5.7_proof}
	In the procedure that was described above, every step consists of $O(1)$ message exchanges within fragments, or between adjacent fragments. By Claim \hyperref[C5.3]{5.3}, every such exchange can be carried out through $O(1)$ communication cycles.
	
	We review additional memory space required by these steps, beyond the needs of maintaining communication cycles. Memory is used for counting base fragments, or counting responses from child fragments. Fragments store their number of subsets. Every fragment also stores the total number of base fragments contained in its  descendants in the virtual tree. Additionally, during an execution of the algorithm, every fragment needs to know the number of its children in the virtual subtree. All these values can be stored using $O(\log n)$ bits.
\end{proof}

If we are solving the MSF problem, terminated fragments \textit{do} participate in slot allocation. If a fragment $F$ has terminated during phase $t$ (or during a previous phase), at the end of phase $t$ its leader vertex $r_F$ requests a slot range, of the same size that $F$ was using during phase $t$.

When we are solving the MSF problem, and some fragment $F$ is terminated during the slot set generation procedure described here, fragment $F$ requests and receives a range within the new set. The leader vertex $r_F$ of a terminated fragment $F$ does not take part in steps 1, 2 and 4 of the process described above. During part 3, the leader vertex sends a range request. The range is meant only for the base fragments that $F$ contains. In step 5, $r_F$ delivers new slots for phase $t+1$ to its base fragment roots, using the slots of phase $t$ for communication. In consequence, at the end of each phase, a new slot set is generated that enables all non-small connected components to participate in communication cycles.

\subsection{Using the Algorithm to Solve Partwise Aggregation}
\label{5.4_RefToPWA}

At the completion of each phase, after the slot set is generated, we are able to perform convergecast and broadcast over every fragment. Thus, we can perform partwise aggregation over every fragment. We combine this fact with the generalized MSF algorithm in order to solve the partwise aggregation problem. Given a partition $\mathcal{P}=\{P_i\}_{i=1}^{i=N}$, a function $f$ and inputs $\{x_v|v\in V\}$, we first find the MSF for the subgraph induced by partition $\mathcal{P}$ (see Equation (\hyperref[eqn_subgraph_weights]{1}); an edge $e=(u,v)$ belongs to $H$ iff $u$ and $v$ belong to the same part $P_i$). Then, using the last slot set produced by the algorithm, we calculate $f$ over the fragments spanning the parts, and broadcast the respective outcomes.
 
 \begin{Claim}
 	\label{Cor5.9_PWA_FOR_MSF}
 	
 	 When the execution of the MSF algorithm is complete, fragments are the minimal spanning trees of the components of $H$. The slot set for these fragments, created in the last executed phase, makes it possible to perform a PWA computation on these connected components. 
 \end{Claim}
 \begin{proof}
 	%\label{Appendix.C_C5.7_PWA_FOR_MSF_proof}
 	By Claim \hyperref[C5.6_SlotGen_correctness]{5.6}, the slot set is correct at the end of every phase. The slot set is therefore correct at the end of the last phase, when fragments are equal to (non-small) connected components. For any given aggregation function $f$, a convergecast cycle can compute the aggregate of $f$ over inputs for each connected component $C$ of the MSF. A broadcast cycle can then announce the result $f(C)$ to the vertices of every connected component $C$. Once non-small components have performed PWA using the latest slot set, small components can perform PWA with convergecast and broadcast over their respective spanning trees within $O(D+\sqrt n)$ rounds.
 \end{proof}
 
 This procedure has the time and message complexity of $O(1)$ communication cycles, and does not incur any memory complexity beyond the size of inputs for the function $f$.

 \begin{corollary}
	\label{C5.9.1_PWA}
	Using the MSF algorithm, we can carry out a partwise aggregation computation for a given partition $\mathcal{P}=\{P_1,...P_k\}$, commutative associative function $f$ and set of inputs $\{x_v|v\in V\}$ held by the vertices in $V$.
\end{corollary}

\begin{proof}
	
	We observe that on a graph $G=(V,E)$, a partition $\mathcal{P} = \{P_i\}_{i=1}^{i=N}$ induces a subgraph $H_P = (V, E_H)$ such that its connected components coincide with the parts of $\mathcal{P}$. We define $E_H$ to be the subset of edges $e\in E, e = (u,v)$ such that endpoints $u,v$ belong to the same part $P_i\in \mathcal{P}$.
	
	We also observe that in order to apply the MSF algorithm to the induced subgraph $H_P$, it is not necessary to store all the edge weights of $E_H$. The weight of an edge $e=(u,v)$ in $H_P$ can be deduced by the endpoints $u, v$ by exchanging their respective part identifiers (see Equation (\hyperref[eqn_subgraph_weights]{1})). At the start of each phase, every vertex $v$ sends a message with its part ID and fragment ID to all its neighbors. Upon receiving these messages from all neighbors, each vertex searches for the MWOE among edges connecting it to a neighbor with the same part ID and a different fragment ID.	
	
	It is therefore possible to apply the MSF algorithm to $H_P$ using only $O(\log n)$ additional bits of memory - the space required by every vertex $v\in V$ to store the label of its part $P_v$, and aggregation input $x_v$.
	
	Given a graph $G=(V,E)$ and partition $\mathcal{P} = \{P_i\}_{i=1}^{i=N}$, in order to carry out a partwise aggregation computation, first construct an MSF for the induced subgraph $H_P$. Next, use the slot set generated in the last phase of the MSF's construction for convergecast and broadcast.
\end{proof}
\subsection{Analysis of Complexity}
\label{5.5_Cplx}

In this section we summarize the complexity analysis of the entire algorithm. We start by analyzing the complexity of a single phase.

\begin{Claim}
\label{C5.8_Cplx_phase}
A single phase has time complexity of $O((D+\sqrt n) \cdot (\log n) \cdot (\log^* n))$ communication rounds, and message complexity of $O(m + n\log n\log^* n)$ messages. The memory complexity of a phase is $O(\log^2 n)$ bits.
\end{Claim}

\begin{proof}
	\label{C5.8_Cplx_phase_proof}
	It takes a single round, and $O(m)$ messages, for every vertex to find its MWOE-candidate. By Claims \hyperref[C5.3]{5.3}, \hyperref[C5.7_slotSet_Cplx]{5.7}, the remainder of the phase has the time and message complexity of $O(\log^* n)$ communication cycles, and memory complexity of $O(\log^2 n)$ bits.
	
	By Claims \hyperref[CB.3_ConvCastTimeCplx]{B.3}, \hyperref[CB.7_BCastCTimeCplx]{B.7}, every communication cycle requires $O((d(T) + K_b)\log n)$ rounds. By Claims \hyperref[CB.4_ConvCastMsgCplx]{B.4}, \hyperref[CB.8_BCastCMsgCplx]{B.8}, a communication cycle requires $O(d(T)\cdot K_b\cdot\log n)$ messages, where $K_b$ is the number of subsets, and $d(T)$ is the depth of the BFS tree. The depth of the BFS tree is $O(D)$ for graph diameter $D$. We choose $K_b$, the number of base fragments, using the same reasoning as in \cite{E20Simple} (see Claim \hyperref[C2.4_E20_MsgCplx]{2.2}). Base fragment diameter is $O(\max\{D,\sqrt n\})$ and the number of base fragments is $K_b=\min\{\frac{n}{D},\sqrt{n}\}$. The overall message complexity for a communication cycle is therefore $O(d(T)\cdot K_b\cdot\log n)=O(n\log n)$.
\end{proof}

In the next claim we summarize the complexity of the second part of the algorithm.

\begin{Claim}
\label{C5.9_Cplx_alg}
Overall, the second part of the algorithm has time complexity of $O((D+\sqrt{n})(\log^2 n)(\log^* n))$ rounds and message complexity of $O(m\log n + n\log^2n\log^*n)$ messages. It has memory complexity of $O(\log^2 n)$ bits.
\end{Claim}
\begin{proof}
	\label{C5.11_Cplx_alg_proof}
	The algorithm constructs MST within $O(\log n)$ phases. We also take into account the complexity of the setup process (see Claim \hyperref[C5.2]{5.2}), which is executed just once. 
	
	In addition to the memory complexity of the algorithm, we consider the memory complexity of storing the MST. We note that a vertex $v$ stores an edge $e$ only if $v$ has detected $e$ as an MWOE for the fragment that $v$ belongs to during some phase. Over the course of the algorithm, any vertex $v$ will undergo $O(\log n)$ phases. A vertex belongs to one fragment in each phase, and so it may learn up to $O(\log n)$ MST-edges overall, resulting in an upper bound of $O(\log^2 n)$ bits on its memory consumption.
\end{proof}

We have proved the following theorem:

\begin{theorem}
	Our distributed deterministic algorithm solves the MST, the MSF and the PWA problems on general n-vertex m-edge graph in the CONGEST model within time $O((D+\sqrt n)(\log^2 n)\log^* n)$, message complexity $O(m\log n + n\log^2 n \cdot \log^* n)$, using $O(\log^2 n)$ bits of memory per vertex. Moreover, while solving MSF problem requires storing the input graph H, solving the PWA problem does not require any additional memory (beyond the input graph).
\end{theorem}

\section*{Acknowledgments}

The authors are grateful to Bernhard Haeupler and Goran Zuzic for helpful discussions about the memory requirements of the shortcut framework. 
\clearpage

\pagenumbering{gobble}

\bibliographystyle{ACM-Reference-Format}
%\bibliography{Bibliography_final}

\begin{thebibliography}{54}

%%% ====================================================================
%%% NOTE TO THE USER: you can override these defaults by providing
%%% customized versions of any of these macros before the \bibliography
%%% command.  Each of them MUST provide its own final punctuation,
%%% except for \shownote{} and \showURL{}.  The latter two
%%% do not use final punctuation, in order to avoid confusing it with
%%% the Web address.
%%%
%%% To suppress output of a particular field, define its macro to expand
%%% to an empty string, or better, \unskip, like this:
%%%
%%% \newcommand{\showURL}[1]{\unskip}   % LaTeX syntax
%%%
%%% \def \showURL #1{\unskip}           % plain TeX syntax
%%%
%%% ====================================================================

\ifx \showCODEN    \undefined \def \showCODEN     #1{\unskip}     \fi
\ifx \showISBNx    \undefined \def \showISBNx     #1{\unskip}     \fi
\ifx \showISBNxiii \undefined \def \showISBNxiii  #1{\unskip}     \fi
\ifx \showISSN     \undefined \def \showISSN      #1{\unskip}     \fi
\ifx \showLCCN     \undefined \def \showLCCN      #1{\unskip}     \fi
\ifx \shownote     \undefined \def \shownote      #1{#1}          \fi
\ifx \showarticletitle \undefined \def \showarticletitle #1{#1}   \fi
\ifx \showURL      \undefined \def \showURL       {\relax}        \fi
% The following commands are used for tagged output and should be
% invisible to TeX
\providecommand\bibfield[2]{#2}
\providecommand\bibinfo[2]{#2}
\providecommand\natexlab[1]{#1}
\providecommand\showeprint[2][]{arXiv:#2}

\bibitem[Awerbuch(1987)]%
        {Awerbuch87}
\bibfield{author}{\bibinfo{person}{Baruch Awerbuch}.}
  \bibinfo{year}{1987}\natexlab{}.
\newblock \showarticletitle{Optimal distributed algorithms for minimum weight
  spanning tree, counting, leader election, and related problems}
  \emph{(\bibinfo{series}{STOC '87})}. \bibinfo{publisher}{Association for
  Computing Machinery}, \bibinfo{pages}{230–240}.
\newblock
\showISBNx{0897912217}
\href{https://doi.org/10.1145/28395.28421}{doi:\nolinkurl{10.1145/28395.28421}}


\bibitem[Awerbuch et~al\mbox{.}(1990a)]%
        {awerbuch1990improved}
\bibfield{author}{\bibinfo{person}{Baruch Awerbuch}, \bibinfo{person}{Amotz
  Bar-Noy}, \bibinfo{person}{Nathan Linial}, {and} \bibinfo{person}{David
  Peleg}.} \bibinfo{year}{1990}\natexlab{a}.
\newblock \showarticletitle{Improved routing strategies with succinct tables}.
\newblock \bibinfo{journal}{\emph{Journal of Algorithms}} \bibinfo{volume}{11},
  \bibinfo{number}{3} (\bibinfo{year}{1990}), \bibinfo{pages}{307--341}.
\newblock


\bibitem[Awerbuch et~al\mbox{.}(1990b)]%
        {awerbuch1990trade}
\bibfield{author}{\bibinfo{person}{Baruch Awerbuch}, \bibinfo{person}{Oded
  Goldreich}, \bibinfo{person}{Ronen Vainish}, {and} \bibinfo{person}{David
  Peleg}.} \bibinfo{year}{1990}\natexlab{b}.
\newblock \showarticletitle{A trade-off between information and communication
  in broadcast protocols}.
\newblock \bibinfo{journal}{\emph{Journal of the ACM (JACM)}}
  \bibinfo{volume}{37}, \bibinfo{number}{2} (\bibinfo{year}{1990}),
  \bibinfo{pages}{238--256}.
\newblock


\bibitem[Awerbuch and Ostrovsky(1994)]%
        {awerbuchOstrovsky1994memory}
\bibfield{author}{\bibinfo{person}{Baruch Awerbuch} {and}
  \bibinfo{person}{Rafail Ostrovsky}.} \bibinfo{year}{1994}\natexlab{}.
\newblock \showarticletitle{Memory-efficient and self-stabilizing network
  reset}. In \bibinfo{booktitle}{\emph{Proceedings of the ACM Symposium on
  Principles of Distributed Computing}}. \bibinfo{pages}{254--263}.
\newblock


\bibitem[Awerbuch and Peleg(1992)]%
        {awerbuch1992routing}
\bibfield{author}{\bibinfo{person}{Baruch Awerbuch} {and}
  \bibinfo{person}{David Peleg}.} \bibinfo{year}{1992}\natexlab{}.
\newblock \showarticletitle{Routing with polynomial communication-space
  trade-off}.
\newblock \bibinfo{journal}{\emph{SIAM Journal on Discrete Mathematics}}
  \bibinfo{volume}{5}, \bibinfo{number}{2} (\bibinfo{year}{1992}),
  \bibinfo{pages}{151--162}.
\newblock


\bibitem[Awerbuch and Varghese(1991)]%
        {awerbuchVarghese91}
\bibfield{author}{\bibinfo{person}{Baruch Awerbuch} {and}
  \bibinfo{person}{George Varghese}.} \bibinfo{year}{1991}\natexlab{}.
\newblock \showarticletitle{Distributed program checking: a paradigm for
  building self-stabilizing distributed protocols}. In
  \bibinfo{booktitle}{\emph{FOCS}}, Vol.~\bibinfo{volume}{91}.
  \bibinfo{pages}{258--267}.
\newblock


\bibitem[Azarmehr et~al\mbox{.}(2025)]%
        {azarmehr2025massively}
\bibfield{author}{\bibinfo{person}{Amir Azarmehr}, \bibinfo{person}{Soheil
  Behnezhad}, \bibinfo{person}{Rajesh Jayaram}, \bibinfo{person}{Jakub
  {\L}{\k{a}}cki}, \bibinfo{person}{Vahab Mirrokni}, {and}
  \bibinfo{person}{Peilin Zhong}.} \bibinfo{year}{2025}\natexlab{}.
\newblock \showarticletitle{Massively parallel minimum spanning tree in general
  metric spaces}. In \bibinfo{booktitle}{\emph{Proceedings of the Annual
  ACM-SIAM Symposium on Discrete Algorithms (SODA)}}. SIAM,
  \bibinfo{pages}{143--174}.
\newblock


\bibitem[Ben~Basat et~al\mbox{.}(2025)]%
        {BCCHLS25}
\bibfield{author}{\bibinfo{person}{Ran Ben~Basat}, \bibinfo{person}{Keren
  Censor-Hillel}, \bibinfo{person}{Yi-Jun Chang}, \bibinfo{person}{Wenchen
  Han}, \bibinfo{person}{Dean Leitersdorf}, {and} \bibinfo{person}{Gregory
  Schwartzman}.} \bibinfo{year}{2025}\natexlab{}.
\newblock \showarticletitle{Bounded Memory in Distributed Networks}. In
  \bibinfo{booktitle}{\emph{Proceedings of the ACM Symposium on Parallelism in
  Algorithms and Architectures}}. \bibinfo{pages}{566--581}.
\newblock


\bibitem[Blin et~al\mbox{.}(2013)]%
        {blin2013fast}
\bibfield{author}{\bibinfo{person}{L{\'e}lia Blin}, \bibinfo{person}{Shlomi
  Dolev}, \bibinfo{person}{Maria~Gradinariu Potop-Butucaru}, {and}
  \bibinfo{person}{Stephane Rovedakis}.} \bibinfo{year}{2013}\natexlab{}.
\newblock \showarticletitle{Fast Self-Stabilizing Minimum Spanning Tree
  Construction Using Compact Nearest Common Ancestor Labeling Scheme}.
\newblock \bibinfo{journal}{\emph{arXiv preprint arXiv:1311.0798}}
  (\bibinfo{year}{2013}).
\newblock


\bibitem[Chin and Ting(1985)]%
        {ChinTing85}
\bibfield{author}{\bibinfo{person}{F. Chin} {and} \bibinfo{person}{H.~F.
  Ting}.} \bibinfo{year}{1985}\natexlab{}.
\newblock \showarticletitle{An almost linear time and O(nlogn+e) Messages
  distributed algorithm for minimum-weight spanning trees}. In
  \bibinfo{booktitle}{\emph{26th Annual Symposium on Foundations of Computer
  Science (sfcs 1985)}}. \bibinfo{pages}{257--266}.
\newblock
\href{https://doi.org/10.1109/SFCS.1985.7}{doi:\nolinkurl{10.1109/SFCS.1985.7}}


\bibitem[Cole and Vishkin(1986)]%
        {CV}
\bibfield{author}{\bibinfo{person}{Richard Cole} {and} \bibinfo{person}{Uzi
  Vishkin}.} \bibinfo{year}{1986}\natexlab{}.
\newblock \showarticletitle{Deterministic coin tossing with applications to
  optimal parallel list ranking}.
\newblock \bibinfo{journal}{\emph{Information and Control}}
  \bibinfo{volume}{70}, \bibinfo{number}{1} (\bibinfo{year}{1986}),
  \bibinfo{pages}{32--53}.
\newblock


\bibitem[Datta et~al\mbox{.}(2008)]%
        {datta2008self}
\bibfield{author}{\bibinfo{person}{Ajoy~K. Datta}, \bibinfo{person}{Lawrence~L.
  Larmore}, {and} \bibinfo{person}{Priyanka Vemula}.}
  \bibinfo{year}{2008}\natexlab{}.
\newblock \showarticletitle{Self-stabilizing leader election in optimal space}.
  In \bibinfo{booktitle}{\emph{Symposium on Self-Stabilizing Systems}}.
  Springer, \bibinfo{pages}{109--123}.
\newblock


\bibitem[Dufoulon et~al\mbox{.}(2022)]%
        {dufoulon2022almost}
\bibfield{author}{\bibinfo{person}{Fabien Dufoulon}, \bibinfo{person}{Shay
  Kutten}, \bibinfo{person}{William~K Moses~Jr}, \bibinfo{person}{Gopal
  Pandurangan}, {and} \bibinfo{person}{David Peleg}.}
  \bibinfo{year}{2022}\natexlab{}.
\newblock \showarticletitle{An almost singularly optimal asynchronous
  distributed MST algorithm}.
\newblock \bibinfo{journal}{\emph{arXiv preprint arXiv:2210.01173}}
  (\bibinfo{year}{2022}).
\newblock


\bibitem[Elkin(2006a)]%
        {elkin2006faster}
\bibfield{author}{\bibinfo{person}{Michael Elkin}.}
  \bibinfo{year}{2006}\natexlab{a}.
\newblock \showarticletitle{A faster distributed protocol for constructing a
  minimum spanning tree}.
\newblock \bibinfo{journal}{\emph{J. Comput. System Sci.}}
  \bibinfo{volume}{72}, \bibinfo{number}{8} (\bibinfo{year}{2006}),
  \bibinfo{pages}{1282--1308}.
\newblock


\bibitem[Elkin(2006b)]%
        {elkin2006unconditional}
\bibfield{author}{\bibinfo{person}{Michael Elkin}.}
  \bibinfo{year}{2006}\natexlab{b}.
\newblock \showarticletitle{An Unconditional Lower Bound on the
  Time-Approximation Trade-off for the Distributed Minimum Spanning Tree
  Problem}.
\newblock \bibinfo{journal}{\emph{SIAM J. Comput.}} \bibinfo{volume}{36},
  \bibinfo{number}{2} (\bibinfo{year}{2006}), \bibinfo{pages}{433--456}.
\newblock


\bibitem[Elkin(2017)]%
        {E20Simple}
\bibfield{author}{\bibinfo{person}{Michael Elkin}.}
  \bibinfo{year}{2017}\natexlab{}.
\newblock \showarticletitle{A Simple Deterministic Distributed MST Algorithm,
  with Near-Optimal Time and Message Complexities}. In
  \bibinfo{booktitle}{\emph{Proceedings of the ACM Symposium on Principles of
  Distributed Computing}}. \bibinfo{pages}{157--163}.
\newblock


\bibitem[Elkin et~al\mbox{.}(2014)]%
        {elkin2014can}
\bibfield{author}{\bibinfo{person}{Michael Elkin}, \bibinfo{person}{Hartmut
  Klauck}, \bibinfo{person}{Danupon Nanongkai}, {and} \bibinfo{person}{Gopal
  Pandurangan}.} \bibinfo{year}{2014}\natexlab{}.
\newblock \showarticletitle{Can quantum communication speed up distributed
  computation?}. In \bibinfo{booktitle}{\emph{Proceedings of the ACM Symposium
  on Principles of Distributed Computing}}. \bibinfo{pages}{166--175}.
\newblock


\bibitem[Elkin and Neiman(2018a)]%
        {elkin2018near}
\bibfield{author}{\bibinfo{person}{Michael Elkin} {and} \bibinfo{person}{Ofer
  Neiman}.} \bibinfo{year}{2018}\natexlab{a}.
\newblock \showarticletitle{Near-optimal distributed routing with low memory}.
  In \bibinfo{booktitle}{\emph{Proceedings of the ACM Symposium on Principles
  of Distributed Computing}}. \bibinfo{pages}{207--216}.
\newblock


\bibitem[Elkin and Neiman(2018b)]%
        {elkin2018efficient}
\bibfield{author}{\bibinfo{person}{Michael Elkin} {and} \bibinfo{person}{Ofer
  Neiman}.} \bibinfo{year}{2018}\natexlab{b}.
\newblock \showarticletitle{On efficient distributed construction of near
  optimal routing schemes}.
\newblock \bibinfo{journal}{\emph{Distributed Computing}} \bibinfo{volume}{31},
  \bibinfo{number}{2} (\bibinfo{year}{2018}), \bibinfo{pages}{119--137}.
\newblock


\bibitem[Emek and Wattenhofer(2013)]%
        {emek2013stone}
\bibfield{author}{\bibinfo{person}{Yuval Emek} {and} \bibinfo{person}{Roger
  Wattenhofer}.} \bibinfo{year}{2013}\natexlab{}.
\newblock \showarticletitle{Stone age distributed computing}. In
  \bibinfo{booktitle}{\emph{Proceedings of the ACM Symposium on Principles of
  Distributed Computing}}. \bibinfo{pages}{137--146}.
\newblock


\bibitem[Faloutsos and Molle(2004)]%
        {FM04}
\bibfield{author}{\bibinfo{person}{Michalis Faloutsos} {and}
  \bibinfo{person}{Mart Molle}.} \bibinfo{year}{2004}\natexlab{}.
\newblock \showarticletitle{A linear-time optimal-message distributed algorithm
  for minimum spanning trees}.
\newblock \bibinfo{journal}{\emph{Distrib. Comput.}} \bibinfo{volume}{17},
  \bibinfo{number}{2} (\bibinfo{date}{Aug.} \bibinfo{year}{2004}),
  \bibinfo{pages}{151–170}.
\newblock
\showISSN{0178-2770}
\href{https://doi.org/10.1007/s00446-004-0107-2}{doi:\nolinkurl{10.1007/s00446-004-0107-2}}


\bibitem[Gafni(1985)]%
        {Gafni85}
\bibfield{author}{\bibinfo{person}{Eli Gafni}.}
  \bibinfo{year}{1985}\natexlab{}.
\newblock \showarticletitle{Improvements in the time complexity of two
  message-optimal election algorithms}. In
  \bibinfo{booktitle}{\emph{Proceedings of the ACM Symposium on Principles of
  Distributed Computing}} (Minaki, Ontario, Canada)
  \emph{(\bibinfo{series}{PODC '85})}. \bibinfo{publisher}{Association for
  Computing Machinery}, \bibinfo{address}{New York, NY, USA},
  \bibinfo{pages}{175–185}.
\newblock
\showISBNx{0897911687}
\href{https://doi.org/10.1145/323596.323612}{doi:\nolinkurl{10.1145/323596.323612}}


\bibitem[Gallager et~al\mbox{.}(1983)]%
        {GHS83}
\bibfield{author}{\bibinfo{person}{Robert~G. Gallager},
  \bibinfo{person}{Pierre~A. Humblet}, {and} \bibinfo{person}{Philip~M.
  Spira}.} \bibinfo{year}{1983}\natexlab{}.
\newblock \showarticletitle{A distributed algorithm for minimum-weight spanning
  trees}.
\newblock \bibinfo{journal}{\emph{ACM Transactions on Programming Languages and
  Systems (TOPLAS)}} \bibinfo{volume}{5}, \bibinfo{number}{1}
  (\bibinfo{year}{1983}), \bibinfo{pages}{66--77}.
\newblock


\bibitem[Garay et~al\mbox{.}(1998)]%
        {gkp}
\bibfield{author}{\bibinfo{person}{Juan~A Garay}, \bibinfo{person}{Shay
  Kutten}, {and} \bibinfo{person}{David Peleg}.}
  \bibinfo{year}{1998}\natexlab{}.
\newblock \showarticletitle{A sublinear time distributed algorithm for
  minimum-weight spanning trees}.
\newblock \bibinfo{journal}{\emph{SIAM J. Comput.}} \bibinfo{volume}{27},
  \bibinfo{number}{1} (\bibinfo{year}{1998}), \bibinfo{pages}{302--316}.
\newblock


\bibitem[Gavoille et~al\mbox{.}(2013)]%
        {gavoille2013communication}
\bibfield{author}{\bibinfo{person}{Cyril Gavoille}, \bibinfo{person}{Christian
  Glacet}, \bibinfo{person}{Nicolas Hanusse}, {and} \bibinfo{person}{David
  Ilcinkas}.} \bibinfo{year}{2013}\natexlab{}.
\newblock \showarticletitle{On the communication complexity of distributed
  name-independent routing schemes}. In \bibinfo{booktitle}{\emph{International
  Symposium on Distributed Computing}}. Springer, \bibinfo{pages}{418--432}.
\newblock


\bibitem[Ghaffari and Haeupler(2016)]%
        {ghaffari2016distributed}
\bibfield{author}{\bibinfo{person}{Mohsen Ghaffari} {and}
  \bibinfo{person}{Bernhard Haeupler}.} \bibinfo{year}{2016}\natexlab{}.
\newblock \showarticletitle{Distributed algorithms for planar networks ii:
  Low-congestion shortcuts, mst, and min-cut}. In
  \bibinfo{booktitle}{\emph{Proceedings of the ACM-SIAM Symposium on Discrete
  Algorithms}}. SIAM, \bibinfo{pages}{202--219}.
\newblock


\bibitem[Ghaffari and Haeupler(2021)]%
        {ghaffari2021low}
\bibfield{author}{\bibinfo{person}{Mohsen Ghaffari} {and}
  \bibinfo{person}{Bernhard Haeupler}.} \bibinfo{year}{2021}\natexlab{}.
\newblock \showarticletitle{Low-congestion shortcuts for graphs excluding dense
  minors}. In \bibinfo{booktitle}{\emph{Proceedings of the ACM Symposium on
  Principles of Distributed Computing}}. \bibinfo{pages}{213--221}.
\newblock


\bibitem[Gmyr and Pandurangan(2018)]%
        {gmyr2018time}
\bibfield{author}{\bibinfo{person}{Robert Gmyr} {and} \bibinfo{person}{Gopal
  Pandurangan}.} \bibinfo{year}{2018}\natexlab{}.
\newblock \showarticletitle{Time-message trade-offs in distributed algorithms}.
\newblock \bibinfo{journal}{\emph{arXiv preprint arXiv:1810.03513}}
  (\bibinfo{year}{2018}).
\newblock


\bibitem[Gupta and Srimani(2003)]%
        {gupta2003self}
\bibfield{author}{\bibinfo{person}{Sandeep~KS Gupta} {and}
  \bibinfo{person}{Pradip~K Srimani}.} \bibinfo{year}{2003}\natexlab{}.
\newblock \showarticletitle{Self-stabilizing multicast protocols for ad hoc
  networks}.
\newblock \bibinfo{journal}{\emph{J. Parallel and Distrib. Comput.}}
  \bibinfo{volume}{63}, \bibinfo{number}{1} (\bibinfo{year}{2003}),
  \bibinfo{pages}{87--96}.
\newblock


\bibitem[Haeupler et~al\mbox{.}(2018)]%
        {haeupler2018round}
\bibfield{author}{\bibinfo{person}{Bernhard Haeupler},
  \bibinfo{person}{D.~Ellis Hershkowitz}, {and} \bibinfo{person}{David Wajc}.}
  \bibinfo{year}{2018}\natexlab{}.
\newblock \showarticletitle{Round-and message-optimal distributed graph
  algorithms}. In \bibinfo{booktitle}{\emph{Proceedings of the ACM Symposium on
  Principles of Distributed Computing}}. \bibinfo{pages}{119--128}.
\newblock


\bibitem[Haeupler et~al\mbox{.}(2016)]%
        {haeupler2016low}
\bibfield{author}{\bibinfo{person}{Bernhard Haeupler}, \bibinfo{person}{Taisuke
  Izumi}, {and} \bibinfo{person}{Goran Zuzic}.}
  \bibinfo{year}{2016}\natexlab{}.
\newblock \showarticletitle{Low-congestion shortcuts without embedding}. In
  \bibinfo{booktitle}{\emph{Proceedings of the ACM Symposium on Principles of
  Distributed Computing}}. \bibinfo{pages}{451--460}.
\newblock


\bibitem[Haeupler and Li(2018)]%
        {haeupler2018faster}
\bibfield{author}{\bibinfo{person}{Bernhard Haeupler} {and}
  \bibinfo{person}{Jason Li}.} \bibinfo{year}{2018}\natexlab{}.
\newblock \showarticletitle{Faster distributed shortest path approximations via
  shortcuts}.
\newblock \bibinfo{journal}{\emph{arXiv preprint arXiv:1802.03671}}
  (\bibinfo{year}{2018}).
\newblock


\bibitem[Higham and Liang(2001)]%
        {higham2001self}
\bibfield{author}{\bibinfo{person}{Lisa Higham} {and} \bibinfo{person}{Zhiying
  Liang}.} \bibinfo{year}{2001}\natexlab{}.
\newblock \showarticletitle{Self-stabilizing minimum spanning tree construction
  on message-passing networks}. In \bibinfo{booktitle}{\emph{International
  Symposium on Distributed Computing}}. Springer, \bibinfo{pages}{194--208}.
\newblock


\bibitem[Khan and Pandurangan(2008)]%
        {khan2008fast}
\bibfield{author}{\bibinfo{person}{Maleq Khan} {and} \bibinfo{person}{Gopal
  Pandurangan}.} \bibinfo{year}{2008}\natexlab{}.
\newblock \showarticletitle{A fast distributed approximation algorithm for
  minimum spanning trees}.
\newblock \bibinfo{journal}{\emph{Distributed Computing}} \bibinfo{volume}{20},
  \bibinfo{number}{6} (\bibinfo{year}{2008}), \bibinfo{pages}{391--402}.
\newblock


\bibitem[King et~al\mbox{.}(2015)]%
        {king2015construction}
\bibfield{author}{\bibinfo{person}{Valerie King}, \bibinfo{person}{Shay
  Kutten}, {and} \bibinfo{person}{Mikkel Thorup}.}
  \bibinfo{year}{2015}\natexlab{}.
\newblock \showarticletitle{Construction and impromptu repair of an MST in a
  distributed network with o (m) communication}. In
  \bibinfo{booktitle}{\emph{Proceedings of the ACM Symposium on Principles of
  Distributed Computing}}. \bibinfo{pages}{71--80}.
\newblock


\bibitem[Kor et~al\mbox{.}(2011)]%
        {kor2011tight}
\bibfield{author}{\bibinfo{person}{Liah Kor}, \bibinfo{person}{Amos Korman},
  {and} \bibinfo{person}{David Peleg}.} \bibinfo{year}{2011}\natexlab{}.
\newblock \showarticletitle{Tight bounds for distributed MST verification}. In
  \bibinfo{booktitle}{\emph{28th International Symposium on Theoretical Aspects
  of Computer Science (STACS 2011)}}. \bibinfo{pages}{69--80}.
\newblock


\bibitem[Korman and Kutten(2007)]%
        {korman2007controller}
\bibfield{author}{\bibinfo{person}{Amos Korman} {and} \bibinfo{person}{Shay
  Kutten}.} \bibinfo{year}{2007}\natexlab{}.
\newblock \showarticletitle{Controller and estimator for dynamic networks}. In
  \bibinfo{booktitle}{\emph{Proceedings of the ACM Symposium on Principles of
  Distributed Computing}}. \bibinfo{pages}{175--184}.
\newblock


\bibitem[Korman et~al\mbox{.}(2015)]%
        {KKM15_SelfStabilizing}
\bibfield{author}{\bibinfo{person}{Amos Korman}, \bibinfo{person}{Shay Kutten},
  {and} \bibinfo{person}{Toshimitsu Masuzawa}.}
  \bibinfo{year}{2015}\natexlab{}.
\newblock \showarticletitle{Fast and compact self-stabilizing verification,
  computation, and fault detection of an MST}.
\newblock \bibinfo{journal}{\emph{Distributed Computing}} \bibinfo{volume}{28},
  \bibinfo{number}{4} (\bibinfo{year}{2015}), \bibinfo{pages}{253--295}.
\newblock


\bibitem[Kutten et~al\mbox{.}(2015)]%
        {kutten2015complexity}
\bibfield{author}{\bibinfo{person}{Shay Kutten}, \bibinfo{person}{Gopal
  Pandurangan}, \bibinfo{person}{David Peleg}, \bibinfo{person}{Peter
  Robinson}, {and} \bibinfo{person}{Amitabh Trehan}.}
  \bibinfo{year}{2015}\natexlab{}.
\newblock \showarticletitle{On the complexity of universal leader election}.
\newblock \bibinfo{journal}{\emph{Journal of the ACM (JACM)}}
  \bibinfo{volume}{62}, \bibinfo{number}{1} (\bibinfo{year}{2015}),
  \bibinfo{pages}{1--27}.
\newblock


\bibitem[Kutten and Peleg(1998)]%
        {KP1998fast}
\bibfield{author}{\bibinfo{person}{Shay Kutten} {and} \bibinfo{person}{David
  Peleg}.} \bibinfo{year}{1998}\natexlab{}.
\newblock \showarticletitle{Fast distributed construction of small
  $k$-dominating sets and applications}.
\newblock \bibinfo{journal}{\emph{Journal of Algorithms}} \bibinfo{volume}{28},
  \bibinfo{number}{1} (\bibinfo{year}{1998}), \bibinfo{pages}{40--66}.
\newblock


\bibitem[Lenzen(2016)]%
        {Lenzen2016}
\bibfield{author}{\bibinfo{person}{Christoph Lenzen}.}
  \bibinfo{year}{2016}\natexlab{}.
\newblock \bibinfo{booktitle}{\emph{Lecture Notes on Theory of Distributed
  Systems}}.
\newblock
\urldef\tempurl%
\url{https://www.mpi-inf.mpg.de/fileadmin/inf/d1/teaching/winter15/tods/ToDS.pdf}
\showURL{%
Retrieved Feb 23,2026 from \tempurl}


\bibitem[Lenzen and Patt-Shamir(2013)]%
        {lenzen2013fast}
\bibfield{author}{\bibinfo{person}{Christoph Lenzen} {and}
  \bibinfo{person}{Boaz Patt-Shamir}.} \bibinfo{year}{2013}\natexlab{}.
\newblock \showarticletitle{Fast routing table construction using small
  messages}. In \bibinfo{booktitle}{\emph{Proceedings of the ACM Symposium on
  Theory of Computing}}. \bibinfo{pages}{381--390}.
\newblock


\bibitem[Lenzen and Patt-Shamir(2015)]%
        {lenzen2015fast}
\bibfield{author}{\bibinfo{person}{Christoph Lenzen} {and}
  \bibinfo{person}{Boaz Patt-Shamir}.} \bibinfo{year}{2015}\natexlab{}.
\newblock \showarticletitle{Fast partial distance estimation and applications}.
  In \bibinfo{booktitle}{\emph{Proceedings of the ACM Symposium on Principles
  of Distributed Computing}}. \bibinfo{pages}{153--162}.
\newblock


\bibitem[Lenzen et~al\mbox{.}(2019)]%
        {lenzen2019distributed}
\bibfield{author}{\bibinfo{person}{Christoph Lenzen}, \bibinfo{person}{Boaz
  Patt-Shamir}, {and} \bibinfo{person}{David Peleg}.}
  \bibinfo{year}{2019}\natexlab{}.
\newblock \showarticletitle{Distributed distance computation and routing with
  small messages}.
\newblock \bibinfo{journal}{\emph{Distributed Computing}} \bibinfo{volume}{32},
  \bibinfo{number}{2} (\bibinfo{year}{2019}), \bibinfo{pages}{133--157}.
\newblock


\bibitem[Mashreghi and King(2021)]%
        {mashreghi2021broadcast}
\bibfield{author}{\bibinfo{person}{Ali Mashreghi} {and}
  \bibinfo{person}{Valerie King}.} \bibinfo{year}{2021}\natexlab{}.
\newblock \showarticletitle{Broadcast and minimum spanning tree with o (m)
  messages in the asynchronous CONGEST model}.
\newblock \bibinfo{journal}{\emph{Distributed Computing}} \bibinfo{volume}{34},
  \bibinfo{number}{4} (\bibinfo{year}{2021}), \bibinfo{pages}{283--299}.
\newblock


\bibitem[Murmu(2015)]%
        {murmu2015distributed}
\bibfield{author}{\bibinfo{person}{Mahendra~Kumar Murmu}.}
  \bibinfo{year}{2015}\natexlab{}.
\newblock \showarticletitle{A distributed approach to construct minimum
  spanning tree in cognitive radio networks}.
\newblock \bibinfo{journal}{\emph{Procedia Computer Science}}
  \bibinfo{volume}{70} (\bibinfo{year}{2015}), \bibinfo{pages}{166--173}.
\newblock


\bibitem[Pandurangan et~al\mbox{.}(2017)]%
        {PRS17}
\bibfield{author}{\bibinfo{person}{Gopal Pandurangan}, \bibinfo{person}{Peter
  Robinson}, {and} \bibinfo{person}{Michele Scquizzato}.}
  \bibinfo{year}{2017}\natexlab{}.
\newblock \showarticletitle{A time-and message-optimal distributed algorithm
  for minimum spanning trees}. In \bibinfo{booktitle}{\emph{Proceedings of the
  ACM SIGACT Symposium on Theory of Computing}}. \bibinfo{pages}{743--756}.
\newblock


\bibitem[Pandurangan et~al\mbox{.}(2018)]%
        {P18SURVEY}
\bibfield{author}{\bibinfo{person}{Gopal Pandurangan}, \bibinfo{person}{Peter
  Robinson}, {and} \bibinfo{person}{Michele Scquizzato}.}
  \bibinfo{year}{2018}\natexlab{}.
\newblock \showarticletitle{The distributed minimum spanning tree problem}.
\newblock \bibinfo{journal}{\emph{Bulletin of EATCS}} \bibinfo{volume}{2},
  \bibinfo{number}{125} (\bibinfo{year}{2018}).
\newblock


\bibitem[Peleg(2000)]%
        {peleg2000distr}
\bibfield{author}{\bibinfo{person}{David Peleg}.}
  \bibinfo{year}{2000}\natexlab{}.
\newblock \bibinfo{booktitle}{\emph{Distributed Computing: A Locality-Sensitive
  Approach}}.
\newblock \bibinfo{publisher}{SIAM}.
\newblock


\bibitem[Peleg and Rubinovich(1999)]%
        {PelegRub99_nearTightLwrBound}
\bibfield{author}{\bibinfo{person}{David Peleg} {and} \bibinfo{person}{Vitaly
  Rubinovich}.} \bibinfo{year}{1999}\natexlab{}.
\newblock \showarticletitle{A near-tight lower bound on the time complexity of
  distributed MST construction}. In \bibinfo{booktitle}{\emph{40th Annual
  Symposium on Foundations of Computer Science (Cat. No. 99CB37039)}}. IEEE,
  \bibinfo{pages}{253--261}.
\newblock


\bibitem[Santoro and Khatib(1985)]%
        {santoro1985labelling}
\bibfield{author}{\bibinfo{person}{Nicola Santoro} {and} \bibinfo{person}{Ramez
  Khatib}.} \bibinfo{year}{1985}\natexlab{}.
\newblock \showarticletitle{Labelling and implicit routing in networks}.
\newblock \bibinfo{journal}{\emph{Comput. J.}} \bibinfo{volume}{28},
  \bibinfo{number}{1} (\bibinfo{year}{1985}), \bibinfo{pages}{5--8}.
\newblock


\bibitem[Sarma et~al\mbox{.}(2012)]%
        {sarmaEtAl12_distributedVerification}
\bibfield{author}{\bibinfo{person}{Atish~Das Sarma}, \bibinfo{person}{Stephan
  Holzer}, \bibinfo{person}{Liah Kor}, \bibinfo{person}{Amos Korman},
  \bibinfo{person}{Danupon Nanongkai}, \bibinfo{person}{Gopal Pandurangan},
  \bibinfo{person}{David Peleg}, {and} \bibinfo{person}{Roger Wattenhofer}.}
  \bibinfo{year}{2012}\natexlab{}.
\newblock \showarticletitle{Distributed verification and hardness of
  distributed approximation}.
\newblock \bibinfo{journal}{\emph{SIAM J. Comput.}} \bibinfo{volume}{41},
  \bibinfo{number}{5} (\bibinfo{year}{2012}), \bibinfo{pages}{1235--1265}.
\newblock


\bibitem[Thorup and Zwick(2001)]%
        {TZ01Compact}
\bibfield{author}{\bibinfo{person}{Mikkel Thorup} {and} \bibinfo{person}{Uri
  Zwick}.} \bibinfo{year}{2001}\natexlab{}.
\newblock \showarticletitle{Compact routing schemes}. In
  \bibinfo{booktitle}{\emph{Proceedings of the ACM Symposium on Parallel
  Algorithms and Architectures}}. \bibinfo{pages}{1--10}.
\newblock


\bibitem[Tixeuil(2009)]%
        {tixeuil2009self}
\bibfield{author}{\bibinfo{person}{S{\'e}bastien Tixeuil}.}
  \bibinfo{year}{2009}\natexlab{}.
\newblock \showarticletitle{Self-stabilizing algorithms}.
\newblock \bibinfo{journal}{\emph{Algorithms and Theory of Computation
  Handbook}} (\bibinfo{year}{2009}), \bibinfo{pages}{26--1}.
\newblock


\end{thebibliography}
%%% -*-BibTeX-*-
%%% Do NOT edit. File created by BibTeX with style
%%% ACM-Reference-Format-Journals [18-Jan-2012].

\pagebreak
\pagenumbering{arabic}
\appendix
\clearpage

\section{Memory Complexity of Existing Algorithms}
\label{Appendix.A_ExistingAlgs_MemCplx}
\subsection{Primitives}
\label{Appendix.A.1_Primitives}
\subsubsection{Trees}
\label{Appendix.A.1.1_Bcast}

Let $G = (V,E)$ be a graph that is storing a distributed representation of subgraph $G'=(V',E')$, so that $ V'\subseteq V$, $ E'\subseteq E$. In particular, $G'$ may be a spanning tree. It is typically assumed that such a representation of $G'$ requires every vertex $v\in V'$ to know, for every edge $e=\{v,u\}$ adjacent to it, whether $e\in E'$. In this paper, we assume that a more compact representation is sufficient - for every edge $e\in E', e=\{v,u\}$, either $u$ or $v$ is required know that $e\in E'$. For a vertex $v'\in V'$, we denote the edges from $E'$ that $v'$ knows as $E^{'}_{v'}$. We argue that this representation is good enough, because it enables $G$ to execute a broadcast over $G'$, with the same asymptotic round and message complexity as a ``full" representation.

\subsubsection{Broadcast}
\label{Appendix.A.1.2_Bcast}
In order to carry out a broadcast over a subtree represented in this way, the graph $G$ first holds a round where every vertex $v$ sends a request message over every adjacent tree edge $e'\in E^{'}_v$. We assume that edges are memory buffers, and that a recipient vertex $u$ can leave such a request in the buffer, ``unopened", and $v$ does not need to re-send it every round. Next, suppose that a vertex $w$ has a message $MSG$ to broadcast. Vertex $w$ iterates over its edges, sending $MSG$ over every edge $e\in E^{'}_{w}$, and also over every edge that holds a request. Every vertex $u\in V$ that receives a broadcast message from some adjacent edge broadcasts that message in the following round.

\subsubsection{Convergecast}
\label{Appendix.A.1.3_Convt}
In a tree subgraph, we require every edge to be known to one of its endpoints. For rooted trees the requirement is more specific. In a rooted tree, we require every edge $e = (u,v)$ between a parent $u$ and a child $v$ to be known to the child $v$. Additionally, in order to carry out a convergecast on a rooted tree, each parent vertex needs know the number of its children. Storage of a general tree \textit{can be} compact, storage of rooted tree is always compact.

In a convergecast calculation, a vertex $v$ waits to receive inputs from all its children. Vertex $v$ aggregates inputs as they arrive, and maintains a count of messages received. Once the count of received messages is equal to the number of children, $v$ knows that the aggregated result is final. In the next round, $v$ sends the aggregate to its parent vertex. 

In some cases (e.g., calculating a minimum) the outcome has a single source vertex $v_m$. Sometimes it is necessary to notify $v_m$ that it is the source vertex, and to send a message to $v_m$ from the root. In order to make it possible to retrace the path back to the source $v_m$, another variable is required. Every vertex $v$ is also required to remember the edge $e_0 = (v, v_0)$ that connects $v$ to the vertex $v_0$ that upcasts the minimum outcome. This requires $O(\log n)$ bits of memory.

We summarize our observations with the following claim, stated without proof:

\begin{Claim}
	\label{CA.1_TreeCplx}
	\leavevmode 
	\begin{enumerate}
		\item A spanning tree can be stored in a distributed fashion with memory complexity $O(\log n)$. A tree stored in this compact fashion is able to carry out convergecast and broadcast algorithms.
		\item Both broadcast and convergecast can be carried out with $O(D)$ time complexity, $O(n)$ message complexity, and $O(\log n)$ local memory complexity.
	\end{enumerate}
\end{Claim}

\subsection {Memory Complexity of GHS and GKP Algorithms}
\label{Appendix.A.2_GHS_and_GKP}
\subsubsection {The GHS pkase}
\label{Appendix.A.2.1_GHS}
 In each phase of the GHS algorithm, two steps occur:
\begin{enumerate}
	\item MST-fragments discovered in the previous phase perform, in parallel, a search for their respective least-weight outgoing edges (MWOEs). The fragments and MWOEs form a virtual graph. This graph is a forest \cite{GHS83}. Every fragment discovers a single MWOE and considers the neighbor across that MWOE its parent in the tree. When two fragments discover the same MWOE, one of them assumes the role of root.
	\item In each of the trees in the forest, the root fragment initiates a broadcast of its ID. The tree thus becomes a single fragment in the next phase.
\end{enumerate}

We review the local memory complexity of a single GHS phase:
\begin{enumerate}
	\item Every vertex sends messages over every edge, with its own fragment ID. It next iterates over its edges and considers their weights and IDs, in order to select the minimal edge whose ID differs from its own. At every step of this process, the vertex is required to hold in its memory the following values: its own ID, the weight and edge ID of the best candidate for MWOE. All these values have size $O(\log n)$ bits.
	
	\item Within every fragment, a convergecast calculation is held for fragment-MWOE. A convergecast calculation has memory complexity $O(\log n)$ (see Section \hyperref[Appendix.A.1.3_Convt]{A.1.3}).
	
	\item The message announcing the chosen MWOE is sent back from the fragment root to the vertex that is adjacent to the selected MWOE. To make this possible, every vertex $v$ participating in the convergecast needs to remember the identity of the neighbor that originated the value that $v$ chose to upcast. This requires additional $O(\log n)$ bits of memory (see Section \hyperref[Appendix.A.1.3_Convt]{A.1.3}).
	
	\item Fragments now form a forest of virtual trees, with MWOEs as edges. A fragment in each virtual tree is designated as root, by the method described by GHS \cite{GHS83}.\footnote{In each subtree, every MWOE is found by one fragment, except a single MWOE that is found by two fragments. One of these two fragments becomes the root, chosen by larger fragment ID.} The MWOE-adjacent vertex in that fragment becomes fragment root, and originates a broadcast of fragment ID over the entire fragment tree. Fragment ID requires $O(\log n)$ bits.
\end{enumerate}

It follows that the memory complexity of the GHS algorithm is $O(\log n)$. We omit the mechanism used by \cite{GHS83} to reduce message complexity to $O(m)$, by sending a query only twice for each edge, once from each endpoint. This would require every vertex $v$ to maintain a record of queried edges. Instead, we assume that edges are checked again at the start of every phase, and message complexity is $O(m\log n)$.

\subsubsection {The GKP phase}
\label{Appendix.A.2.2_GKP}
We review the memory complexity of the additional steps used in the GKP phase.

\begin{enumerate}
	\item After finding MWOEs, in the process of coloring and color reduction, for every fragment, its vertices store its color and the colors of its virtual-tree neighbors. These are $O(\log n)$-size values. 
	\item When generating a maximum matching,  every fragment is required to store its matching status, and the identity of its matched partner. This has memory complexity of $O(\log n)$ bits as well. Once the matching is complete, unmatched fragments store the identity of the fragment they attach to, using another $O(\log n)$ bits. This information is sufficient for maintaining virtual subtrees.
	\end{enumerate}

\subsubsection {Lenzen's Phase}
\label{Appendix.A.2.3_Lenzen}

To implement the version of the GKP phase described by Lenzen \cite{Lenzen2016}, fragments are required to initiate the process of MWOE detection in phase $i$ only if their diameter is at most $2^i$. To apply this policy, in every fragment $F$, at the start of every phase $i$, the root $r_F$ is required to discover the diameter of $F$ and announce via broadcast to the vertices of $F$ whether or not to detect an MWOE in that phase. This additional step requires every vertex $v\in V$ to store its partial value in its fragment' diameter calculation, and a bit indicating participation. This step thus has memory complexity of $O(\log n)$.

\subsection{The GKP Pipeline Procedure}
\label{Appendix.A.3_GKP_pipeline}
Over the course of the Pipeline Procedure, used by \cite{gkp,KP1998fast}, vertices use an auxiliary BFS tree $T$ to calculate the edges still missing from the MST after the first part of the GKP algorithm. Every vertex upcasts cross-fragment edges over $T$. Vertices upcast edges known to them initially, as well as edges received from $T$-children, ordered by weight - least-weight edges are sent first. A vertex $v$ upcasts edges only if the edges do not create a cycle with the subset of the fragment tree that $v$ has upcast in previous rounds. Vertices begin upcasting according to their depth in $T$, in a way that enables efficient pipelining and also ensures that all MST edges are upcast first, and non-MST edges are filtered out.

In order to detect cycles, and filter out non-MST edges, every vertex needs to maintain a representation of the subtree that it has already upcast. This representation is required to describe $O(\sqrt{n})$ virtual fragments and edges, together with edge weights. Thus, the representation may require up to $O(\sqrt{n})$ local memory space.

\subsection{Message-Efficient Algorithms}
\label{Appendix.A.4_MSGEfficient}
The algorithm in \cite{E20Simple} also proceeds in phases. The algorithm consists of base fragments performing partial calculations to find the MWOE of the fragment that contains them. Each base fragment upcasts its MWOE-candidate to the root vertex $v_r$ of a BFS tree. The root vertex then determines which candidate is the MWOE for each fragment, and notifies every base fragment of the new ID ot the fragment. In order to determine MWOEs, vertex $v_r$ is required to retain several items of information that incur memory complexity beyond polylogarithmic. 

First, while MWOE candidates are being sent, $v_r$ needs to retain an MWOE-candidate for each fragment, for up to $O(\sqrt n)$ fragments in any phase. Secondly, $v_r$ messages every base fragment with its new fragment ID once MWOEs are found, and so needs to track the state of every base fragment - its old fragment ID as well as its new fragment ID. As there are $O(\sqrt n)$ base fragments, this also has memory complexity $O(\sqrt n)$.

See also the discussion in Section \hyperref[1.3_TechOvViw]{1.3} concerning the memory complexity of the algorithm of \cite{haeupler2018round}.

\subsection{Compact Routing}
\label{2.5_Compact_routing}

In our algorithm we make use of a routing scheme to direct messages to their destinations. Specifically, we make use of a tree routing scheme. We use the compact tree routing scheme devised by Thorup and Zwick \cite{TZ01Compact} (TZ). Here we briefly define routing schemes, describe the TZ tree routing scheme, and state the memory complexity of its elements without proof. 

\subsubsection{Routing Schemes}
\label{A.5.1_Routing_schemes}

A routing scheme for a graph $G=(V,E)$ is a distributed algorithm that can deliver a message from any vertex to any other vertex. Suppose that vertex $v \in V$ is in the process of executing the routing scheme algorithm. Whenever $v$ receives a message meant for destination vertex $u$, vertex $v$ decides if the message has reached its final destination and, if not, decides over which edge to forward the message next (see \cite{peleg2000distr}, Ch. 9). The routing scheme relies on a local \textit{routing table} stored by each vertex, and on \textit{routing labels} associated with each destination vertex. 

In the general case, when a message is sent from an origin vertex $v_{orig}$ to some destination $v_{dest}$, $v_{orig}$ is required to know the routing label. Based on the routing label and its local routing table, $v_{orig}$ determines the next edge $e_{next} = (v_{orig},v_{next})$ in the path of the message to $v_{dest}$, and it also creates a \textit{header} to send to $v_{next}$ alongside the message. The header is then used by $v_{next}$ to determine the next step and generate the next header, and so on. In the TZ tree routing algorithm, headers are simply identical to the routing label.

A routing scheme consists of:
\begin{enumerate}
	\item A function that, given a routing label or header, determines the next edge and next header for a message being routed.
	\item An algorithm for computing routing tables and routing labels.
\end{enumerate}

The memory complexity of the routing scheme depends on the sizes of routing tables, and on the sizes of routing labels. We will assume that every vertex knows the label associated with it. 

\subsubsection{Thorup-Zwick's Compact Tree Routing}
\label{A.5.2_TZ_Tree_Routing}
Next, we describe the TZ routing scheme, and state its memory complexity, as well as the round and message complexity needed to compute its routing tables and routing labels. 

The TZ routing label consists, for a vertex $v\in V$, of the list of edges leading from the tree root to $v$. The size of this list is reduced from $O(D)$ edges to $O(\log n)$ edges by omitting certain edges from the path. An edge leading from $u$ to $v$ is omitted if the subtree rooted at $v$ is of size at least half the size of the subtree rooted at $u$. Every vertex $v\in V$ has at most one such edge, and the routing table for $v$ indicates this edge. The routing tables and labels contain some additional information. Specifically, every vertex has a number, and every subtree has an associated interval. A label $lbl_v$ for vertex $v$ indicates $v$'s number, and the routing table at $v$ indicates the interval for the subtree rooted at $v$. This additional information has size $O(\log n)$ for labels and tables.
\begin{Claim}
	\label{CA.1_TZ_clpx}

	\begin{enumerate}
		\item The size of a routing label is $O(\log^2 n)$ bits, and the size of a routing table is $O(\log n)$ bits.\footnote{Thorup and Zwick \cite{TZ01Compact} provide a routing scheme with label size $O(\log n)$. However, it relies on the ability of a vertex to rename adjacent edges to assist the scheme. If this ability is not available, explicitly simulating edge renaming in vertex memory would by itself impose a prohibitively large memory cost (of $O(deg(v))$).}
		\item The time complexity of the setup algorithm in the CONGEST model is $O(D \log n)$, and the message complexity is $O(n \log n)$.
	\end{enumerate}
	
\end{Claim}

%\clearpage
\section{Analysis of the Communication Cycle}
\label{Appendix.B_Analysis}

\subsection{The Convergecast Cycle}
\label{B.1_analysis}

\subsubsection{Avoiding Congestion}
\label{3.3.1_DepthSprBtwFrags}
We show that if all subset roots use Equation (\hyperref[eqn_slot_to_depth]{2}) to choose the round on which they insert their partial outcomes, and if all vertices aggregate and upcast inputs immediately, then in every round, the convergecast calculation for each fragment takes place at a different depth. Thus, it is possible to carry out each convergecast successfully, as if it was the only calculation being performed. As a result, there is no need to handle congestion and queuing, or take action to distinguish between values from different fragments. Depth separation enables us to establish the correctness and the round-, message- and vertex memory complexity of the communication cycle. 

\begin{Claim}
	\label{CB.1_CycleSeparatesFrags}
	If the following conditions hold:
	\begin{enumerate}
		\item At every round $t$, every vertex $v$ upcasts to its $T$-parent an $f$-aggregate of all the inputs it has received in round $t$.
		\item If a vertex $v$ itself belongs to $F_i$, and has an input for $F_i$'s convergecast, $v$ adds the input to the calculation of $f$-aggregate in round $t = i+d(T)-d(v)$, i.e., $d(v) = dp(i,t)$.
	\end{enumerate}
	Then:		
	\begin{enumerate}
		\item Any vertex $v$ will receive upcast messages concerning the fragment $F_i$ during at most one round - round $t = i+d(T)-d(v)$, i.e., the round $t$ when $d(v) = dp(i,t)$.
		\item The vertex $v$ will upcast at most one message concerning $F_i$. This happens during round $t = i+d(T)-d(v)$, i.e., the round $t$ when $d(v) = dp(i,t)$.
	\end{enumerate}
\end{Claim}

\begin{proof}
	We prove the claim by induction on the depth $d(T)$ of tree $T$. If $d(T)=0$, then $T$ contains only the root vertex $r_T$, and the claim is true because no messages are sent. Next, consider a tree of depth $d(T)$. The vertices of depth $d(T)$ are leaves, and thus they never receive $f$-aggregate values. They upcast only $f$-aggregates over their own inputs. If such a vertex belongs to $F_i$, it will only send a message during round $i+d(T)-d(v)=i$. Hence, for vertices of depth $d(T)$, the assertion of the claim holds. 
	
	Now consider the subtree that remains after excluding leaves of depth $d(T)$ from the initial tree. Its depth is $d(T)-1$. The induction assumption holds for any messages coming from within it, leaving only messages originating from the leaves. For any fragment $F_i$, messages from the $d(T)$-depth leaf vertices reach the $(d(T)-1)$-depth vertices at round $i+d(T)-(d(T)-1)=i+1$. Round $i+1$ is the round that we have already designated for $(d(T)-1)$-depth vertices to receive and send messages from fragment $F_i$.
\end{proof}

\subsubsection{Correctness}
\label{3.3.2_correctness}

The following claim establishes the correctness of a single convergecast communication cycle. 
\begin{Claim}
	\label{CB.2_ConvCastCorrectness}
	After a convergecast cycle is performed, for each fragment $F_i$, the root vertex $r_i$ of the fragment holds a value $f_i$ equal to the outcome of a convergecast calculation of the function $f$ performed over all input values initially held by the vertices of $F_i$, and only over these values.
\end{Claim}

\begin{proof}
	
	The intrinsic step of convergecast is held over disjoint subsets. It is immediate that when the intrinsic step completes, all subset roots hold the convergecast outcomes of their subsets. 
	
	In the extrinsic step, every partial result initially held by a subset root is added to the calculation of its fragment, because for every subset root $r_B\in F_i$ at depth $d(r_B)$, there exists a round $t$ such that $dp(i,t) =  d(r_B)$. During this round $t$, the root vertex $r_B$ handles the calculation for its own fragment $F_i$. 
	
	By Claim \hyperref[CB.1_CycleSeparatesFrags]{B.1}, partial results from $F_i$ are only aggregated with other partial results from $F_i$. Hence, the outcome which reaches the root $r_T$ at any round $t=d(T)+i, i\in[1,K]$, is aggregated precisely over the set of $F_i$'s inputs.
	
	Each outcome $f_i$, after reaching the BFS tree root $r_T$ during round $d(T)+i$, is routed to the fragment leader vertex $r_i$ within the next $d(T)$ rounds. It will reach leader vertex $r_i$ as long as the routing label is correct.
	
\end{proof}

\subsubsection{Analysis of Complexity}
\label{3.3.3_complexity}
The intrinsic step of a communication cycle - convergecast calculations occurring in parallel over a forest of disjoint trees that span subsets - requires up to $D_b$ rounds. 

\begin{Claim}
	\label{CB.3_ConvCastTimeCplx}
	The extrinsic step requires $O((d(T) + K)\log n)$ rounds.
\end{Claim}

\begin{proof}
	For a given fragment $F_i$, its window of ownership visits vertex depths from $d(T)$ (on round $i$) to $0$ (on round ($i + d(T)$)) during rounds $i,\ldots i+d(T)$. Within another $d(T)$ rounds, the outcome reaches the designated fragment leader $r_i$. Within $2d(T) + K$ rounds, the extrinsic step is therefore completed for all fragments. Because routing label size is $O(\log^2 n)$, these are $2d(T) + K$ ``inflated'' rounds, or $O((d(T) + K)\log n)$ actual rounds.
\end{proof}

Message complexity for the intrinsic step is $O(n)$, as it consists of convergecasts over a forest of disjoint trees spanning subsets of graph $G$. 

\begin{Claim}
	\label{CB.4_ConvCastMsgCplx}
	The upper bound for message complexity in the extrinsic step is $O((d(T)\cdot K_b)\log n)$.
\end{Claim}

\begin{proof}
	Each of the $K_b$ subsets originates a single message. Every such message will be sent over up to $2d(T)$ edges before aggregating with another message or reaching its destination. As each message is of size $O(\log^2 n)$, we multiply by a factor of $\log n$.
\end{proof}

%\subsubsection{Memory Complexity}
%\label{3.3.3_memCplx}
Finally, we analyze the memory complexity of the convergecast communication cycle.
\begin{Claim}
	\label{CB.5_ConvCastMemCplx}
	The memory complexity of a single communication cycle is $O(\log^2 n)$.
\end{Claim}

\begin{proof}
	In order to implement the communication cycle procedure, it is enough for each vertex to store the following variables, whose sizes are all $O(\log^2 n)$. 
	
	\begin{itemize}
		\item The share of vertex $v$ in the tree data structures that make up BFS-tree $T$ and subset tree $T_b$ (i.e., the identity of the parent edge, and the number of children).
		
		\item The depth $d(v)$ of vertex $v$ in BFS tree $T$, and the overall depth $d(T)$.
		
		\item Every vertex $v\in F_i$ is required to store its slot number $i$.
		
		\item To enable tree routing, every vertex $v \in V$ stores a routing table of size $O(\log n)$. Each vertex is also required to know its own routing label, and the routing label of its fragment leader $r_i$. These routing labels require $O(\log^2 n)$ bits \cite{TZ01Compact}.
		
	\end{itemize}
	
	In the course of the convergecast calculation, a vertex $v$ may be required to store up to two partial results at any given moment. Firstly, $v$ is required to store its own partial result, for the fragment that $v$ itself belongs to. Secondly, for round $t$, vertex $v$ is required to store the partial result that it is calculating during this round for fragment $F_i$, where $i = d(v) - d(T) + t$. At the end of round $t$, vertex $v$ no longer needs to retain the partial outcome for $F_i$ (See Claim \hyperref[CB.1_CycleSeparatesFrags]{B.1}). The size of these partial results is $O(\log n)$. The partial result for each fragment $F_i$ is accompanied by a routing label of size $O(\log^2 n)$. 
\end{proof}

\begin{Remark}
\label{B.1.Remark_Convt_MemCplx_wrt_storingPartialResults}
We note that in general, convergecast with $O(\log^2 n)$-size messages might require up to $O(deg(v)\log^2 n)$ memory space at a vertex $v\in V$. In a convergecast with $O(\log n)$ size messages, a vertex $v\in V$ receiving $O(deg(v))$ messages can iterate through them, discarding each message after aggregating it. In a convergecast with asymptotically larger messages, the vertex $v$ might need to locally store partial copies of each message throughout the $O(\log n)$ rounds that are needed for these messages to arrive. In the particular case presented in this work, however, the problem does not arise, because all $O(\log^2 n)$-size messages consist of an $O(\log n)$-size input for convergecast, and an $O(\log^2 n)$-size routing label. An ``inflated'' round will consist of a single round where the vertex $v$ iterates over convergecast inputs, and $O(\log n)$ subsequent rounds, where the vertex $v$ collects the appropriate label. Vertex $v$ will always be able to determine which label to collect once the convergecast output is decided. The label will either be identical for all inputs, or it will be the label associated with the input that the vertex $v$ elected to upcast.
\end{Remark}

\subsection{The Broadcast Cycle}
\label{4.3_Broadcast Cycle Analysis}

\subsubsection{Correctness}
\label{4.3.1_Correctness}
To establish the correctness of the extrinsic step, we prove that for each fragment $F_i, i\in [K]$, when its leader $r_i$ sends a message, the message reaches the roots of the subsets of $F_i$. As we have shown in Claim \hyperref[CB.1_CycleSeparatesFrags]{B.1}, messages reach the BFS-root in the order of their slot numbers. Using consecutive p-q slot sets guarantees that for every subset $B\subseteq F_i$, the query message from $B$ will arrive after the p-message from $F_i$, and before the p-message from any other fragment. Thus, the reply to $B$ will be composed while the local message variable held by $r_T$ equals $msg_i$. Correctness follows as long as the slot allocation meets the requirements - every p-q range is consecutive and does not overlap with other p-q ranges, and the p-slot belongs to the leader vertex $r_i$.

\begin{Claim}
	\label{CB.6_BCastCorrectness}
	Suppose that a broadcast cycle takes place, following the algorithm described in Section \hyperref[4.2_pq_broadcast]{4.2}. Suppose also that the set of slots assigned to fragments and subsets meets the following requirements: 
	
	\begin{itemize}
		\item Every fragment $F_i$ consisting of $k$ subsets $B_1,\ldots,B_k$, has a unique consecutive range of $k+1$ slots that does not overlap with other ranges.
		\item Within every such range, the first slot is assigned to the fragment $F_i$, and the remaining $k$ slots are assigned to the $k$ different subsets of $F_i$.
	\end{itemize}
	
	Then, after executing a broadcast communication cycle, for each fragment $F_i$, each vertex $v\in F_i$ receives the message $msg_i$ sent by $F_i$'s leader vertex $r_i$.
\end{Claim}

\begin{proof}
	The root vertex $r_T$ replies to each query $\langle Q, lbl_b\rangle$ sent by a subset $B$ by returning to label $lbl_b$ the last message that $r_T$ has received. Since the p-q ranges are all consecutive and non-overlapping, for every q-slot $j_B$ of subset $B$, the most recent p-slot preceding $j_B$ is the p-slot of $B$'s fragment leader. By Claim \hyperref[CB.1_CycleSeparatesFrags]{B.1}, the upcast messages reach the root in the order of the numbers of their associated slots. That is, if the message upcast in slot $j$ arrives in round $t$, then the message from slot $j + 1$ arrives in round $t+1$. Thus, for every query from subset $B$ of $F_i$, the message sent from $r_T$ in response to the query will be the message published by the leader $r_i$ of $F_i$. 
	
	Once all queries are met with replies, the root of every subset $B\subseteq F_i$ has the message $msg_i$ from the leader $r_i$ of $F_i$. Next, subset roots perform a broadcast inside their respective subsets. Since every vertex $v\in F_i$ is a part of some subset $B\in \mathcal{B}_i$, every vertex $v$ receives $msg_i$.
\end{proof}
\subsubsection{Analysis of Complexity}
\label{4.3.2_Cplx}

Next, we analyze the time, message and memory complexity of the broadcast communication cycle. We start with time complexity. 

\begin{Claim}
	\label{CB.7_BCastCTimeCplx}
	The broadcast cycle has time complexity of $O((d(T)+K_b)\log n + D_b)$ rounds.
\end{Claim}
\begin{proof}
	For the extrinsic step of the broadcast cycle, calculation of time complexity is similar to the convergecast cycle. Within $d(T)+K+K_b$ rounds all messages reach $r_T$, and within another $d(T)$ rounds all replies from $r_T$ reach their respective destinations. We note that every fragment $F$ is partitioned into one or more subsets, and thus $K \leq K_b$. Taking into consideration that the message size is $O(\log^2 n)$, we get time complexity $O((d(T)+K_b)\log n)$.
	
	The intrinsic step consists of broadcasts over a forest of disjoint trees with depth $O(D_b)$. Messages in the intrinsic step do not include routing labels, and so have size $O(\log n)$. The intrinsic step therefore requires $O(D_b)$ rounds.
\end{proof}

Next, we analyze the message complexity.

\begin{Claim}
	\label{CB.8_BCastCMsgCplx}
	The broadcast cycle has message complexity of $O((d(T)\cdot K_b)\log n + n)$ messages.
\end{Claim}
\begin{proof}
	For the extrinsic step, the overall number of vertices initiating messages is $K+K_b\leq2K_b$, including both fragments and subsets. Each message takes up to $2d(T)$ steps until it is delivered to its final destination. The total number of messages sent is $O(d(T)\cdot K_b)$. When we take into account that routing label size is $O(\log^2 n)$, the number of messages is multiplied by a factor of $O(\log n)$, and so we have $O((d(T)\cdot K_b)\log n)$ messages.
	The intrinsic step requires $O(n)$ messages.
\end{proof}

Finally, we analyze the memory complexity of the broadcast cycle.

%\subsubsection{Memory Complexity}
%\label{4.3.3_MemCplx}
\begin{Claim}
	\label{CB.9_BCastCMemoryCplx}
	The memory complexity of the broadcast cycle is \\$O(\log^2 n)$ bits.
\end{Claim}
\begin{proof}
	By Claim \hyperref[CB.5_ConvCastMemCplx]{B.5}, the memory complexity of the convergecast cycle is $O(\log^2 n)$ bits. The implementation of the broadcast cycle requires individual vertices to remember, besides the variables used in the convergecast cycle, one more variable. This is the message variable, hosted by the BFS tree root $r_T$. The root $r_T$ stores up to one message variable at any given round. The size of this variable is $O(\log n)$. 
\end{proof}

%\clearpage

%\begin{Remark}
%	\label{Remark_SlotGen_MSF}
%
%\end{Remark}

%We summarize the main properties of this procedure, its correctness and complexity, in Claims \hyperref[C5.5]{5.5} and \hyperref[C5.6]{5.6}, respectively. We give the proofs below.
 
%Proof for Claim \hyperref[C5.5]{5.5}:

%Proof for Claim \hyperref[C5.6]{5.6}:

%\clearpage

%\section{Pseudocode Algorithms}
%\label{Appendix.G_Algs}
%\subsection{A Single Convergecast Cycle}

%\subsection{Interval Allocation}

%\clearpage
%\section{Figures}
%\label{Appendix.F_Figs}
%\subsection{The Progression of Time Slots}

%\subsection{Range Propagation}

%\subsection{Step 2 of the Slot Set Generation Procedure}

%\pagebreak

\end{document}